\documentclass{informs3}
\RequirePackage{tgtermes}
\RequirePackage{newtxtext}

\RequirePackage{newtxmath}

\OneAndAHalfSpacedXI %

\usepackage{endnotes}
\let\footnote=\endnote
\let\enotesize=\normalsize
\def\notesname{Endnotes}%
\def\makeenmark{$^{\theenmark}$}
\def\enoteformat{\rightskip0pt\leftskip0pt\parindent=1.75em
	\leavevmode\llap{\theenmark.\enskip}}

\usepackage{enumerate} %
\usepackage{bbm} %
\usepackage{bm} %
\usepackage{stmaryrd} %
\usepackage{xcolor} %
\usepackage{nicefrac, xfrac} %
\usepackage[normalem]{ulem} %
\usepackage{hyperref}
\hypersetup{colorlinks=true,linkcolor=blue,citecolor=blue,urlcolor=blue}
\usepackage{url}
\usepackage{comment}

\newcommand{\R}{\mathbb{R}} 
 
\newcommand{\Z}{\mathbb{Z}}
\newcommand{\N}{ \llbracket1,n\rrbracket }

\newcommand{\la}{\lambda}
\newcommand{\mt}{\mu^*}

\DeclareMathAlphabet{\mymathbb}{U}{BOONDOX-ds}{m}{n}

\newcommand{\qandq}{\quad \mbox{and} \quad}

\newcommand{\qforq}{\quad\mbox{for }}

\newcommand{\qforallq}{\quad\mbox{for all }}

\DeclareMathAlphabet{\mymathbb}{U}{BOONDOX-ds}{m}{n}
\newcommand{\deq}{\stackrel{\rm d}{=}}

\usepackage{graphicx}
\usepackage{multirow}
\usepackage{hhline}

\usepackage[small, margin=1cm]{caption}
\usepackage{subcaption} %

\usepackage{appendix}
\usepackage{color}
\usepackage{cancel}
\definecolor{strcolor}{rgb}{0.6, 0.2, 0.6}
\definecolor{commentcolor}{rgb}{0.3125, 0.5, 0.3125}
\definecolor{keycol}{rgb}{0, 0, 1}

\usepackage{bbm}

\usepackage{listings}
\usepackage{hyperref}
\usepackage{url}

\newcommand {\bea}{\begin{eqnarray}}
	\newcommand {\eea}{\end{eqnarray}}

\def\blot{\quad \mbox{$\vcenter{ \vbox{ \hrule height.4pt
				\hbox{\vrule width.4pt height.9ex \kern.9ex \vrule width.4pt}
				\hrule height.4pt}}$}}

\usepackage{natbib}
\bibpunct[, ]{(}{)}{,}{a}{}{,}%
\def\bibfont{\small}%

\TheoremsNumberedThrough     %
\ECRepeatTheorems

\EquationsNumberedThrough    %

\gdef\AQ#1{}
\gdef\CQ#1{}

\hypersetup{
  pdftitle={Stability of Fork-Join Systems with Redundancy and Heterogeneous Servers},
  pdfauthor={Chutong Gao; Seyed Iravani; Ohad Perry},
  pdfsubject={Reject and Resubmit, Operations Research. Resubmitted on July 26, 2026},
  pdfkeywords={fork-join systems; redundancy; stochastic network stability; capacity allocation}
}

\begin{document}

\RUNAUTHOR{Gao, Iravani, and Perry}

\RUNTITLE{Fork-Join Systems with Redundancy}

\TITLE{Stability of Fork-Join Systems with Redundancy and Heterogeneous Servers}

\ARTICLEAUTHORS{%
\AUTHOR{Chutong Gao}
\AFF{Department of Industrial Engineering and Management Sciences, Northwestern University, Evanston, IL 60208, \EMAIL{chutong@u.northwestern.edu}} %
\AUTHOR{Seyed Iravani}
\AFF{Department of Industrial Engineering and Management Sciences, Northwestern University, Evanston, IL 60208, \EMAIL{s-iravani@northwestern.edu}} %
\AUTHOR{Ohad Perry}
\AFF{Department of Operations Research and Engineering Management, Southern Methodist University, Dallas, TX 75205, \EMAIL{operry@smu.edu}} %
} %

\ABSTRACT{
We consider the stability problem of fork-join systems with redundancy (FJR) and heterogeneous servers under both static and dynamic capacity-allocation policies. In an $(n,k)$ FJR system, each arriving job is split into $n$ independent tasks, with one task assigned to each of $n$ parallel servers. Once $k \le n$ tasks have been processed, they are joined and the corresponding job departs the system; the remaining $n-k$ unprocessed tasks are then removed and are therefore termed \emph{redundant}.

We first identify the nominal traffic intensity and characterize the maximal stability region, defined as the set of traffic intensities for which there exists an admissible policy that stabilizes the system. We then establish conditions under which this maximal stability region is attained for two classes of policies: static and dynamic. Specifically, we show that for static allocation policies, in which service capacities remain fixed over time, maximality is achieved whenever the fastest server is allocated no more than $1/k$ of the total service capacity. For dynamic allocation policies, in which a fixed total service capacity may be repeatedly reallocated among the servers, we show that maximality is achieved whenever the cumulative capacity allocated to the $j$ shortest queues does not exceed $j/k$ of the total capacity for every $j=1,\ldots,k-1$.

Our analysis is based on a projection of the $(n,k)$ FJR system onto a simpler $(k,k)$ system that has no redundancy, together with a novel sample-path comparison argument for multidimensional processes based on the generalized Schur-convex order.
}


\renewcommand{\HISTORYname}{\textit{Current status}:\enskip}
\HISTORY{\textit{Reject and Resubmit, Operations Research. Resubmitted on July 26, 2026}}

\maketitle

\section{Introduction}\label{section:introduction}

We consider fork-join systems with redundancy (FJR) and $n \ge 2$ heterogeneous servers. In these systems, each arriving job is split by a \emph{fork} operation into $n$ tasks, with each task routed to the queue of a different station. Tasks are processed according to the first-come-first-served (FCFS) discipline at their respective stations and, upon service completion, move to a join-queue associated with the station at which they were processed.
A job departs the system once any $k$ of its tasks, where $1<k\le n$, have completed service. Specifically, when a task completes service and enters its join-queue, a \emph{join} operation is triggered if $k-1$ tasks from the same job are already waiting in other join-queues. In that case, the $k$ completed tasks are joined and the job departs the system. Otherwise, the newly completed task remains in its join-queue until the $k$th task of the same job completes service. Upon the departure of a job, the remaining $n-k$ tasks are deemed \emph{redundant} and are immediately removed from the system without being processed to completion.

We denote by $(n,k)$ an FJR system with $n$ stations---each consisting of a queue, a server, and a join-queue---in which a job departs once any $k\le n$ of its tasks have completed service. A $(3,2)$ system is depicted in Figure \ref{fig:fork_join_redundancy}.

\begin{figure}[htbp]
    \centering
    \includegraphics[width = 0.8\textwidth]{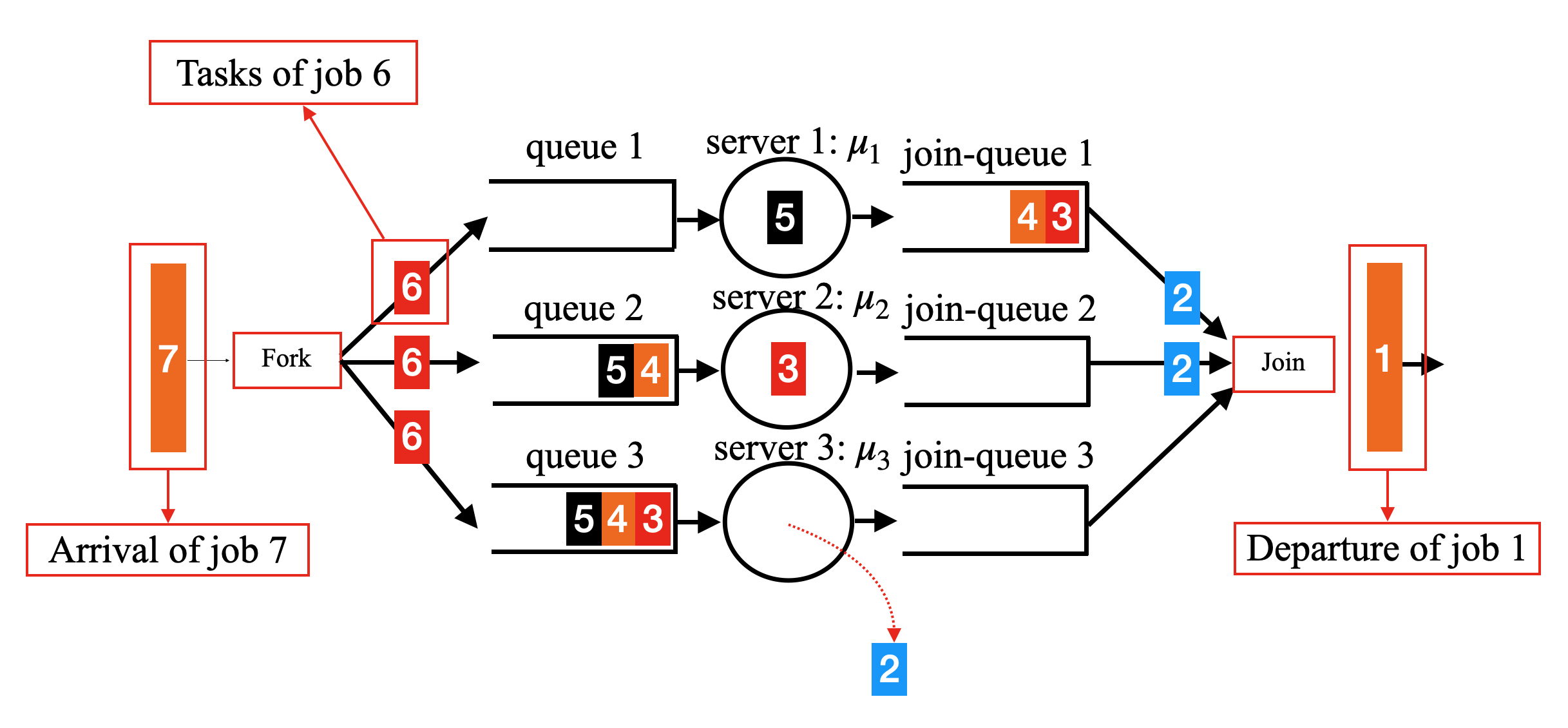}
    \caption{Illustration of a $(3,2)$ system. The second completed task of job 2 has just entered its join-buffer, triggering a join operation that collects the two completed tasks of job 2 from join-queues 1 and 2, after which the job departs the system. Simultaneously, the remaining redundant task of job 2 is removed from station 3.} \label{fig:fork_join_redundancy}
\end{figure}

FJR systems arise naturally in a wide range of applications, including distributed storage and cloud computing (\cite{rizk2016stochastic,joshi2017efficient, wang2019delay, harchol2021open}), healthcare operations (\cite{gardner2016queueing, gardner2017redundancy, ozkan2019control, ozkan2022control, carmeli2023state}), manufacturing systems (\cite{schol2022large,meijer2024optimization}), and parallelized machine-learning tasks (\cite{lee2017speeding,li2020federated, hu2021distributed, ouyang2022training, anthropic_agents_2024, dai2024deepseekmoe}). Depending on the application, a task may represent either a distinct component of a job or a full replica of the job itself.

For example, in cloud storage systems, a file (job) is encoded into $n$ blocks (tasks), and retrieving any $k\le n$ blocks suffices to reconstruct the file (\cite{joshi2012coding, joshi2014delay, joshi2015queues, joshi2017efficient, lee2017mds}). %
In contrast, in redundant parallel large-language-model (LLM) inference, an inference request consisting of a single prompt (job) is replicated across $n$ independent model executions, and the request is completed once responses from any $k$ of the $n$ executions have been obtained. At this point, the remaining unfinished executions are considered redundant and are terminated.
 (\cite{anthropic_agents_2024, agarwal2025first, kim2026atropos}).
The queueing framework studied here captures both interpretations.

\paragraph{\textbf{Heterogeneity in fork-join systems.}}
Server heterogeneity is intrinsic to many applications of fork-join systems. For example, in cloud storage and computing systems, service capacity is often expanded incrementally, either by adding new stations or by increasing the capacity of existing ones (\cite{gardner2019smart}). Likewise, when tasks correspond to different components of a job, their service requirements may vary across stations. Such scenarios are not captured by most of the existing literature, which typically assumes that tasks have independent and identically distributed (i.i.d.) sizes and are processed at identical service rates (cf.~\cite{joshi2017efficient}).

\paragraph{\textbf{Dynamic capacity allocation.}}
In some applications of fork-join systems, a fixed total service capacity can be reallocated dynamically between the different servers. Examples include cross-trained workers who can assist multiple stations in emergency departments  (\cite{ozkan2019control, ozkan2022control}), shared power budgets that can be redistributed among computers (\cite{pedarsani2014robust, pedarsani2017robust}), and
\textit{flexible and collaborative} servers (\cite{bassamboo2012little}, \cite[Chapter 2.1]{dai2020processing}) that can jointly process a single task with additive service capacity.

In this paper, we study FJR systems with heterogeneous servers under both fixed and dynamically allocated service capacities. The standard setting of homogeneous servers processing i.i.d.\ tasks arises as a special case of our results.

\subsection{Stability of FJR Systems} \label{subsection:maximalStabilityRegion}

Despite their practical importance and the extensive academic attention they have received, FJR systems remain poorly understood because of their inherent complexity. Indeed, in addition to the high dimensionality of the state space, the redundancy mechanism introduces a significant analytical complication by rendering the observable state of the system incomplete: some tasks will end up being redundant, yet their identities cannot be determined from the current system state.

As a result, even the most fundamental questions become highly nontrivial. In particular, the characterization of the nominal traffic intensity and the stability region as functions of the arrival and service rates remains unknown. These questions are open even for Markovian systems with homogeneous servers and become considerably more difficult in the presence of heterogeneous service capacities or dynamically controlled servers, as considered in this paper.

We address the preceding questions by first identifying the parameter $\rho := k\lambda/\mt$ as the nominal traffic intensity, where $\lambda$ denotes the job arrival rate and $\mt$ is the total service capacity, i.e., the sum of the service rates across all stations. We then characterize the \textit{maximal stability region}, namely, the set of traffic intensities for which there exists a capacity-allocation policy that stabilizes the system, and show that this region is the interval $[0,1)$. In particular, there exists a capacity-allocation policy that stabilizes the system if and only if $\rho < 1$.

Having characterized the maximal stability region, the next natural questions concern how an FJR system should be designed and, when dynamic capacity allocation is available, controlled so as to attain this maximal stability region. We address both the design and the control problems separately. 

\paragraph{\textbf{Maximal design.}} Under static capacity-allocation policies, the total service capacity is distributed among the $n$ servers and remains fixed over time. Consequently, a static policy corresponds to the design of an $(n,k)$ system with a prescribed total service capacity $\mt$. In this setting, the most fundamental question is that of \textit{maximal design}, namely, how the total service capacity should be allocated among the $n$ servers in order for the system to achieve the maximal stability region.

\paragraph{\textbf{Maximal control.}}
The dynamic-allocation setting presents control problems, the most fundamental of which is characterizing \textit{maximal controls}, namely, dynamic capacity-allocation policies under which the maximal stability region is achieved.  

\subsection{Contributions}\label{subsection:contribution}
In the context of the fundamental research problems discussed above, our contribution is fourfold.
\begin{enumerate}
     \item \textit{\textbf{Traffic intensity and maximal stability region.}} We identify the traffic intensity as $\rho := k\lambda / \mt$, and prove that the maximal stability region is $[0,1)$. In particular, we show that the system is unstable under every admissible policy whenever $\rho \ge 1$, whereas there exist admissible policies that stabilize the system whenever $\rho < 1$.
    
	\item \textit{\textbf{Static allocation.}} We prove that any static capacity allocation for which the service rate of the fastest server is no larger than $\mt/k$ is maximal. Since the service rate of the fastest server necessarily lies in the interval $[\mt/n,\mt/k]$, the presence of redundancy enlarges this interval and thereby increases the robustness of static allocations in achieving the maximal stability region. In particular, in a $(k,k)$ system, which has no redundancy, the unique maximal design is the one that assigns identical service rates to all $k$ servers.
    
	\item \textit{\textbf{Dynamic allocation.}} We quantify the greater robustness of dynamic capacity-allocation policies relative to static policies by proving that a dynamic policy is maximal if for each $j=1,\ldots,k-1$, the total service capacity allocated at any time $t$ to the $j$ shortest queues at that time $t$ does not exceed $\mt j/k$.
    
	 \item \textit{\textbf{Analytical contribution.}} In addition to our contributions to practice and to the FJR literature, we also make analytical contributions that extend beyond the specific setting considered here. First, in the context of FJR systems, we show that each $(n,k)$ system can be projected onto an induced $(k,k)$ system, which reduces the state and action spaces and overcomes the analytical difficulty introduced by redundancy.
     
     Second, we introduce a novel sample-path stochastic order based on the generalized Schur-convex (GSC) order, which facilitates direct comparison between the sample paths of multidimensional stochastic processes via coupling arguments. Specifically, the GSC order transforms the problem of proving stochastic dominance between multidimensional processes into a collection of one-dimensional comparisons involving sums of their ordered components; see \S \ref{subsection:kkproperty}.
\end{enumerate}

\subsection{Notation}

Throughout the paper, lowercase letters denote scalars, e.g., $x$; bold lowercase letters denote vectors, e.g., $\bm x$; uppercase letters denote one-dimensional stochastic processes, e.g., $X$; and bold uppercase letters denote vector-valued stochastic processes, e.g., $\bm X$. 

We let $\Z$ and $\R$ denote the sets of integers and real numbers, respectively, with $\Z_+ := \Z \cap [0,\infty)$ and $\R_+ := [0,\infty)$. For $a,b \in \Z_+$, let $\llbracket a, b\rrbracket := \Z_+ \cap [a,b]$. For a set $\mathbb A$, we write $\mathbb A^k$ for the set of $k$-dimensional vectors whose components take values in $\mathbb A$. The $0$th norm of a vector $\bm{x}\in\R^n$, namely, the number of non-zero components of $\bm x$, is denoted by
$
\|\bm{x}\|_0 = \sum_{i=1}^n \mathbbm{1}{(x_i \neq 0)},
$
where $\mathbbm{1}(\mathbb A)$ is the indicator function of the set $\mathbb A$. We write $\bm x^+$ for the componentwise positive part of a real-valued vector $\bm x$. We denote by $\bm e$ the vector whose components are all $1$, and by $\bm e_i$ the $i$th unit vector, whose $i$th component is equal to $1$ and all other components are $0$.
For $\bm x\in\Z_+^k$, $\mathcal R(\bm x)$ denotes the vector obtained by rearranging the components of $\bm x$ in nondecreasing order, breaking ties according to the original index order.
We denote by $\mathcal R(\Z_+^k)$ the set of $\Z_+^k$-valued vectors with non-decreasing component values, namely,
$
\mathcal R(\Z_+^k) := \{\bm x\in \Z_+^k: x_1 \le x_2 \le \cdots \le x_k\}.
$

For two real-valued stochastic processes $X$ and $Y$, we write $X \le_{\text{st}} Y$ if there exist processes $\mathcal{X}$ and $\mathcal{Y}$, defined on the same probability space, such that $\mathcal{X}(t)\le \mathcal{Y}(t)$ w.p.1 for all $t\ge 0$, and $X \deq \mathcal{X}$, $Y \deq \mathcal{Y}$. %

Let $\eta:=\{\eta(t): t\ge0\}$ denote the constant process satisfying $\eta(t)=1$ for all $t\ge0$. For $\bm b\in\mathbb R^k$, we write $\bm b\eta$ for the $\mathbb R^k$-valued constant process defined by $\bm b \eta(t) = \bm b$ for all $t\ge 0$.   
Finally, for a continuous-time Markov chain (CTMC) $\bm Q$ we write
	$\bm q\to \bm q' \text{ at rate }\lambda(\bm q, \bm q')$
if $\bm Q$ can make a direct transition from state $\bm q$ to state $\bm q'$ at a rate $\lambda(\bm q, \bm q')$.

\subsection{Organization}
The remainder of the paper is organized as follows. After a literature review in \S\ref{section:literature}, we introduce the model in detail in \S\ref{section:model}. In \S\ref{section:induce}, we introduce the induced $(k,k)$ system, obtained by projecting an $(n,k)$ system onto a lower-dimensional space, and develop the GSC order. We employ the induced $(k,k)$ system and the GSC order in \S\ref{section:static} and \S\ref{section: dynamic} to study the stability of $(n,k)$ systems under static and dynamic policies, respectively. We conclude with a summary in \S\ref{section:summary}. Finally, some proofs, in addition to auxiliary results and their proofs, appear in \S\ref{app:proofs}.

\section{Literature Review} \label{section:literature}

 \paragraph{\textbf{FJR systems.}}
We begin with a review of the literature on $(n,d,k)$ systems, where $n \ge d \ge k$. In these systems, each arriving job is split into $d$ tasks, which are assigned to $d$ out of the $n$ parallel stations according to a prescribed routing policy. A job departs once any $k$ of its $d$ tasks have completed service, at which point the remaining $d-k$ tasks are discarded. When $d=k=1$, the model reduces to a parallel-server system in which each job is routed to one of $n$ dedicated servers. Even in this simpler setting, stability can depend delicately on the routing policy, as demonstrated in \cite{moyal2022stability}.

The FJR systems studied in this paper correspond to the $(n,n,k)$ special case of the $(n,d,k)$ model, which we denote by $(n,k)$. These systems exhibit \textit{full redundancy}, since each job sends one task to every station. Full redundancy has been extensively studied in the literature and arises naturally in many applications; see, e.g., \cite{shah2015redundant, joshi2014delay, joshi2015queues, joshi2017efficient, lee2017speeding, lee2017mds}. In particular, \cite{shah2015redundant} show that, when service times are i.i.d.\ and either exponentially distributed or heavier tailed than exponential, the mean steady-state sojourn time in an $(n,k)$ system is minimized among all $(n,d,k)$ systems with $d<n$. Likewise, \cite{joshi2014delay} use an $(n,d,k)$ model with i.i.d.\ exponential service times to study distributed cloud-storage systems and demonstrate that full redundancy can substantially reduce download latency; see also \cite{joshi2015queues, lee2017mds}.

\cite{gardner2017redundancy} perform an exact analysis of the $(n,d,1)$ system with Poisson job arrivals, homogeneous servers, i.i.d.\ exponential service times, and a uniform-at-random task-dispatching policy. They derive closed-form expressions for both the mean sojourn time and the stability region. Their analysis suggests that the mean sojourn time is strictly decreasing in~$d$, implying that full redundancy minimizes the mean sojourn time.

A more refined characterization of the relationship between the service-time distribution and the benefits of redundancy is provided by \cite{joshi2017efficient} for the $(n,d,1)$ system. In particular, they prove that when service times are log-convex, full redundancy achieves the minimal mean sojourn time.

Regarding fork-join systems with no redundancy, \cite{flatto1984two} and \cite{nelson1988approximate} characterize the stability region of the $(2,2,2)$ system and derive the generating function of its stationary queue length distribution. 
For $(n,n,n)$ systems with $n>2$, only bounds on the mean sojourn time were established; see \cite{baccelli1989fork, varki1999mean, ko2008sojourn, thomasian2014analysis}, and the survey paper \cite{sethuraman2022analysis}. \cite{dai2020processing} consider a fork-join network consisting of multiple nested $(n,n,n)$ systems, which is treated as a stochastic processing network (SPN). %

It is significant that none of the papers cited above established a stability condition for the $(n,k)$ system with full redundancy, not even when the servers are homogeneous, except for the simple $(n,n,1)$ case, which is a special case of the $(n,d,1)$ systems studied in \cite{gardner2017redundancy, anton2021stability}.
 
\paragraph{\textbf{Flexible capacity allocation.}}
In the context of flexible capacity allocation, our paper contributes to the literature on flexible and collaborative servers, which studies queueing systems with servers that can be assigned to different queues and collaborate on processing a single task. See (\cite{down2006dynamic, bassamboo2012little}) for applications in parallel queueing systems; (\cite{andradottir2005throughput, down2006dynamic}) for applications in tandem queues; and (\cite{andradottir2003dynamic, pedarsani2014robust, pedarsani2014scheduling, dai2020processing}) for more general queueing networks. 

However, only a limited number of works consider flexible capacity allocation in fork-join networks, and none incorporates redundancy.
\cite{pedarsani2014scheduling} study the allocation of flexible and collaborative servers in a queueing network composed of multiple nested $(n,n,n)$ queues, and formulate a linear program for numerically computing a throughput-maximizing policy.
\cite{marin2016dynamic} analyze a saturated $(n,n,n)$ system with flexible and collaborative server allocation, and prove that the saturated system is unstable.
\cite{ozkan2019control} consider an $(n,d,d)$ system in which each queue has a dedicated server, together with a subset of queues that share a single flexible server. They formulate and solve the scheduling problem for the flexible server via an approximating Brownian control problem, and establish asymptotic optimality of the resulting policy under a stochastic-order optimality criterion. \cite{ozkan2022control} extends these results to systems with multiple flexible servers.

\section{The Model}\label{section:model}

We now describe the $(n,k)$ system in detail: The system consists of $n$ parallel stations, each comprising a queue (buffer), in which tasks wait for service, a single server, and a \textit{join-queue}, in which completed tasks wait until the $k$th task belonging to the same job has completed service, at which point the job departs the system, and all of its tasks (either in the queues or in the join-queues) are removed. Jobs arrive according to a Poisson process with rate $\la$. Upon arrival, each job is split into $n$ tasks, with one task routed to each station.

We assume that task sizes are i.i.d.\ exponential random variables with mean $1$, and that each server processes tasks at a rate equal to its assigned service capacity. (We use the terms ``capacity'' and ``service rate'' interchangeably, both referring to the linear rate at which a task's remaining workload is reduced.) More precisely, if server $i$ is assigned capacity $\mu_i$ over a time interval $[s,t)$, then, whenever tasks are present at station~$i$, workload is processed at constant rate $\mu_i$ throughout that interval.
(Example~\ref{exNH} below illustrates how the model also accommodates non-identically distributed task sizes.)
The total service capacity across all servers is fixed at $\mt>0$, which, without loss of generality, we normalize to $1$. %

The service discipline at each station is FCFS and non-idling. Upon completion of service, a task moves to the join-queue associated with the same station. A job departs the system as soon as any $k$ of its $n$ tasks have completed service. At that moment, the completed tasks are joined, while the remaining $n-k$ tasks of the same job are immediately removed from the stations at which they are still present. Note that the case $n=k$ corresponds to a system without redundancy.

For $t \ge 0$ and $i \in \llbracket 1, n\rrbracket$, let $Z(t)$ be the number of jobs in the system at time $t$; $Q_i(t)$ be the number of tasks in queue $i$, including the task in service; and $V_i(t)$ be the number of tasks waiting in join-queue $i$ at time $t \ge 0$. We write
$$\bm Q(t) :=\bigl(Q_1(t),\ldots,Q_n(t) \bigr) \quad \text{and} \quad \bm V(t) := \bigl( V_1(t),\ldots,V_n(t) \bigr),$$ 
and omit the time argument from the notation when referring to the corresponding stochastic processes. For example, $\bm Q := \{\bm Q(t) : t \ge 0\}$.

\subsection{Capacity-Allocation Policies} 

Under a static capacity allocation, the total service capacity $\mt=1$ is distributed among the $n$ servers, and the capacity assigned to each server remains fixed over time. In contrast, under a dynamic capacity-allocation policy $\pi$, the service capacity assigned to server~$i$ at time $t$, denoted by $\Gamma_i^\pi(t)$, is time-dependent. The corresponding capacity-allocation process is denoted by $\bm{\Gamma}^\pi := \{\bm{\Gamma}^\pi(t): t\ge 0\}$, where $\bm \Gamma^\pi$ takes values in
\begin{equation*}
	\Delta^n := \left\{\bm{d}\in \R^n_+: \sum_{i=1}^n d_i = 1\right\}.
\end{equation*}

\begin{remark}
	When $\mt \ne 1$, $\Gamma_i^\pi(t)$ represents the proportion of the total capacity allocated to server $i$, so that the instantaneous service rate at any time $t\ge0$ is $\Gamma_i^\pi(t)\mt$, for all $i \in \llbracket 1, n \rrbracket$. 
\end{remark}

A policy is said to be admissible if decisions are made at event epochs---namely, at job-arrival and task-completion times---and depend only on the state of the system at the corresponding decision epochs; see \S\ref{subsection:stateDescriptor} for a rigorous definition. We denote by $\Pi^{(n,k)}$ the class of all admissible policies for $(n,k)$ systems, and append a superscript $\pi$ to processes, e.g., $\bm Q^{\pi}$ and $\bm V^{\pi}$, whenever we wish to emphasize their dependence on a particular policy $\pi$.

The class of admissible policies encompasses a broad range of realistic settings, as illustrated by the following examples.

\begin{example}[static allocation] \label{exHomoServers}
    For $\bm d \in \Delta^n$, let $\pi_{\bm d}$ denote the policy satisfying $\bm{\Gamma}^{\pi_{\bm d}}(t) \equiv \bm d$ for all $t \ge 0$. Under $\pi_{\bm d}$, station~$i$ operates at the constant service rate~$d_i$. The special case in which $d_i = 1/n$ for all $i \in \llbracket 1,n \rrbracket$ corresponds to a system with statistically homogeneous servers, each operating at rate~$1/n$. Hence, static capacity allocation, and in particular, the classical homogeneous-server setting, arises as a special case of the dynamic-allocation framework.
\end{example} 

\begin{example}[non-identically-distributed tasks and heterogeneous servers] \label{exNH}
    An $(n,k)$ system in which task sizes in queue $i$ are exponentially distributed with rate $\nu_i$, and server $i$ processes work at a fixed rate $\mu_i$, can be represented within our framework by taking the total service capacity to be $\mt := \sum_{i=1}^n \nu_i \mu_i$, and by exercising the policy $\pi_{nh}$, defined via $\Gamma_i^{\pi_{nh}}(t) \equiv \nu_i \mu_i / \mt$ for all $t\ge 0$ and for all $i\in \N$. 
\end{example} 

\begin{example}[flexible and collaborative servers]\label{exFlexibleServers}
In this setting, there are $m \ge 1$ servers (possibly with $m>n$), with different fixed capacities $\mu_j$, $j\in \llbracket1,m\rrbracket$, and each server may be assigned to any of the $n$ stations. The servers are flexible and collaborative in the sense that they may move between stations and, when several servers are assigned to the same station, their capacities combine additively to determine the processing rate of the task in service. 

Such collaborative-server models are widely studied in the literature on dynamic server allocation; see, e.g., \cite{bassamboo2012little, pedarsani2014robust}. 
An important contemporary example is the aforementioned redundant multi-agent LLM inference systems. In this setting, the server pool is often composed of servers with different fixed capacities due to varied hardware generations and hardware configurations. For example, a single server can be formed by $1$, $8$, or $72$ GPUs, respectively (\cite{nvidia_dgx_superpod_2026}).
We refer to an FJR system with flexible and collaborative servers as an $(n,k)$-FC system.

Formally, let $\mathcal{P}=(\mathcal{P}_1,\mathcal{P}_2, \cdots, \mathcal{P}_n)$ be a partition of $\llbracket 1,m\rrbracket$, so that $\cup_{i=1}^n \mathcal P_i = \llbracket 1,m \rrbracket$ and $\mathcal P_i \cap \mathcal P_j = \varnothing$, for $i \ne j$, where $\varnothing$ denotes the empty set. Note that some of the $\mathcal{P}_i$'s may be empty. Let $ \bm{\mu}:= \{\mu_1, \mu_2, \cdots, \mu_m \}$ be the collection of service rates. A capacity-allocation policy $\pi_{FC}$ in an $(n,k)$-FC system selects actions in the set
\begin{equation} \label{eqDelta_u}
	\Delta^n_{\bm{\mu}}:=\left\{ \bm{d}\in \Delta^n: \text{there exists } \mathcal{P} \subseteq \llbracket 1,m\rrbracket ~\text{ s.t. }~d_i= \frac{ \sum_{j\in \mathcal{P}_i} \mu_j}{ \sum_{\ell=1}^n \mu_{\ell}}, i\in \llbracket 1, n \rrbracket \right\}.
\end{equation}
Since $\Delta^n_{\bm{\mu}}\subsetneq \Delta^n$, every such policy is admissible for the $(n,k)$ systems under our consideration. Hence $(n,k)$-FC systems are a subclass of the $(n,k)$ systems studied here.
\end{example}

\paragraph{\textbf{The class of static policies.}}
As demonstrated in Example \ref{exHomoServers}, static policies are a special subclass of admissible policies. Specifically, $\pi\in\Pi^{(n,k)}$ is static if $\bm{\Gamma}^\pi(t) \equiv \bm{d}$ for all $t\ge 0$ and for some $\bm{d}\in \Delta^n$; we denote that policy by $\pi_{\bm d}$. The class of all static policies is denoted by
\begin{equation*}
\hat{\Pi}^{(n,k)}:= \left\{ \pi_{\bm{d}} \in \Pi^{(n,k)}: \bm{d}\in \Delta^n \right\}.
\end{equation*}

\subsection{A CTMC Representation} \label{subsection:stateDescriptor}

We now show that $\bm Q$ evolves as a CTMC and characterize its transition rates. First, note that FCFS service for tasks at each station implies that jobs depart globally in FCFS order. Indeed, suppose Job B arrives after Job A. By the time Job B departs, exactly $k$ stations must have completed processing Job B’s tasks. Since each station operates under FCFS, any station that has processed a task of Job B must already have processed the corresponding task of Job A. Hence, Job A must have departed before Job B.

Next, index the jobs present at time $t$ by $1,\dots,Z(t)$ according to their order of arrival, where job $Z(t)$ is the youngest job, i.e., the last job to arrive prior to time $t$. Let $\tau_{j,i}(t)$ denote the task of job $j$ assigned to station $i$, for $j=1,\dots,Z(t)$. Observe that each station contains exactly $Z(t)$ tasks. Under FCFS service, the join queue at station $i$, whose length is $V_i(t)$, consists of the tasks $\tau_{1,i}(t),\ldots,\tau_{V_i(t),i}(t)$, whereas the queue at station $i$, including the task currently in service, consists of $\tau_{V_i(t)+1,i}(t),\ldots,\tau_{Z(t),i}(t)$. Consequently,
\begin{equation}\label{eq:zqv}
Z(t)=Q_i(t)+V_i(t),\qforallq i\in \llbracket 1, n \rrbracket.
\end{equation}

For $\bm q \in \Z^n_+$, let $\bm \sigma(\bm q) := \big(\sigma_1(\bm q), \dots, \sigma_n(\bm q) \big)$ be the permutation of the indices $\llbracket 1, n \rrbracket$ that orders elements of $\bm q$ in nondecreasing order, with ties broken according to the original index order. We drop the argument $\bm q$ when the context is clear, e.g., we write $q_{\sigma_i} := q_{\sigma_i(\bm q)}$. Then for $\bm \sigma(\bm Q(t))$ we have %
\begin{equation*}
	Q_{\sigma_1} (t) \le \cdots \le Q_{\sigma_n} (t), \text{ which together with \eqref{eq:zqv}, imply that } V_{\sigma_1} (t) \ge \cdots \ge V_{\sigma_n} (t).
\end{equation*}

Clearly, any nonempty join-queue $i$ 
necessarily contains task $\tau_{1,i}(t)$ from the oldest job in the system. 
Since this job departs as soon as $k$ of its tasks are processed, at most $k-1$ join-queues can be nonempty. Thus \begin{equation}\label{eq:join-empty-largest}
V_{\sigma_k}(t)=V_{\sigma_{k+1}}(t)=\cdots=V_{\sigma_n}(t)=0,
\end{equation}
which together with \eqref{eq:zqv} yields
\begin{equation}\label{eq:maxq}
Q_{\sigma_k}(t)=Q_{\sigma_{k+1}}(t)=\cdots=Q_{\sigma_n}(t)=Z(t)=\max_{1\le i\le n} Q_i(t).
\end{equation}
Hence,
\begin{equation*}
V_i(t) = Z(t) - Q_i(t) = \max_{1\le j\le n} Q_j(t)-Q_i(t) ,\qquad i\in \llbracket 1, n \rrbracket.
\end{equation*}

It follows that, given $\bm Q(t)$, the values of both $Z(t)$ and $\bm V(t)$ are determined uniquely via the functions $z: \R^n_+ \to \R_+$ and $\bm v: \R^n_+ \to \R^n_+$, respectively, where
\begin{equation*}
z(\bm q):=\max_{1\le i\le n} q_i,
\quad \text{and} \quad
\bm v(\bm q):=z(\bm q) \bm e- \bm q.
\end{equation*}
Therefore, $\bm Q(t)$ is a $\Z^n_+$-valued CTMC on the state space
\begin{equation} \label{StateSpace}
	\mathbb{S}^{(n,k)} := \{\bm q \in \Z_+^n: q_{\sigma_k} = q_{\sigma_{k+1}} = \cdots = q_{\sigma_{n}}\}.
\end{equation}
In particular, 
$
\|\bm v(\bm q)\|_0 \le k-1 \text{ for any }\bm q\in \mathbb S^{(n,k)}.
$

\paragraph{\textbf{Capacity-allocation policies.}}
An admissible dynamic policy $\pi \in \Pi^{(n,k)}$ can be represented by a map $\bm \gamma^\pi: \mathbb S^{(n,k)} \rightarrow \Delta^n$, where for each $\bm q \in \mathbb S^{(n,k)}$, the vector
\begin{equation*}
\bm\gamma^\pi(\bm q)=\bigl(\gamma_1^\pi(\bm q),\ldots,\gamma_n^\pi(\bm q)\bigr)\in\Delta^n
\end{equation*}
specifies the service capacity $\gamma^\pi_i(\bm q)$ assigned to server $i$ when the system is in state $\bm q$. Thus, the capacity-allocation process $\bm \Gamma^\pi$ satisfies $\bm \Gamma^\pi(t) = \bm \gamma^\pi(\bm Q(t))$.
Note that for $\bm d \in \Delta^n$ and a static policy $\pi_{\bm d}$, 
    $\bm \gamma^{\pi_{\bm d}}(\bm q) \equiv \bm d$ %
for all $\bm q \in \mathbb{S}^{(n,k)}$, so that $\bm \Gamma^{\pi_{\bm d}}(t) \equiv \bm d$ for all $t \ge 0$ w.p.1.

The next proposition summarizes the transition rates of the CTMC $\bm Q$.
\begin{proposition}\label{prop:q-transitions}
For a given $\bm q\in\mathbb S^{(n,k)}$ and $\pi\in\Pi^{(n,k)}$, let $\bm v = \bm v(\bm q)$ and $\bm \sigma = \bm \sigma(\bm q)$. The transitions of $\bm Q$ from state $\bm q$ are:
\begin{align}
&\bm q\to \bm q+\bm e \quad \text{at rate }\lambda. \label{eq:q-arrival}
\end{align}
If $\|\bm v\|_0<k-1$, then
\begin{align}
&\bm q\to (\bm q-\bm e_{\sigma_i})^+ \quad \text{at rate }\gamma_{\sigma_i}^\pi(\bm q),\quad i \in \llbracket 1, n \rrbracket. \label{eq:q-departure-nojoin}
\end{align}
If $\|\bm v \|_0=k-1$, then
\begin{align}
&\bm q\to (\bm q-\bm e_{\sigma_i})^+ \quad \quad \quad ~~\text{at rate }\gamma_{\sigma_i}^\pi(\bm q),\quad i\in \llbracket 1, k-1 \rrbracket; \label{eq:q-departure-short}\\
&\bm q\to \left(\bm q-\sum_{j=k}^n \bm e_{\sigma_j}\right)^+ \quad \text{at rate } \sum_{j=k}^n \gamma_{\sigma_j}^\pi(\bm q). \label{eq:q-departure-long}
\end{align}
\end{proposition}

The proof is given in \S \ref{app:q-transitions}.
The transition in \eqref{eq:q-arrival} corresponds to an arrival of a job; the transitions in \eqref{eq:q-departure-nojoin} and \eqref{eq:q-departure-short} correspond to the event of a task completion that does not trigger the removal of redundant tasks. Finally, the transition in \eqref{eq:q-departure-long} corresponds to the event of a task completion that causes its job to depart, and in turn, triggers the removal of that job's redundant tasks.  

\subsection{The Nominal Traffic Intensity and the Maximal Stability Region}\label{subsection:stability}

Since the average amount of workload brought to the system by a job arrival depends on the average processing capacity provided to redundant tasks (that do not finish their processing), it is not immediately clear what the nominal traffic intensity is. 
 
Let
\begin{equation} \label{rho}
	\rho:= \frac{k \lambda}{\mt} = k \lambda, 
\end{equation}
where the second equality follows from our assumption that $\mt = 1$, and let $\bm Q^{\pi,\rho}$ denote the queue process in an $(n,k)$ system with a given $\rho$ in \eqref{rho} that operates under policy $\pi$. Define
\begin{equation*}
\mathcal S^{(n,k)}(\pi) := \{\rho \in \mathbb{R}_+: \bm{Q}^{\pi,\rho} \text{ is positive recurrent} \}.
\end{equation*}
Due to Theorem \ref{thm:maximal-stability} below, we refer to $\rho$ as the {\it traffic intensity} and to $\mathcal S^{(n,k)}(\pi)$ as the {\it stability region} under a given policy $\pi$. We further denote by $\bar{\mathcal S}^{(n,k)}$ the {\it maximal stability region}, namely the set of traffic intensities for which there exists an admissible policy under which the system is stable. We thus say that a policy $\pi$ is \textit{maximal} if $\mathcal S^{(n,k)}(\pi) = \bar{\mathcal S}^{(n,k)} = [0,1)$, where the last equality holds due to the following theorem, whose proof is given in \S \ref{subsection:maximalStability}.

\begin{theorem}\label{thm:maximal-stability}
	The maximal stability region of an $(n,k)$ system is $\bar{\mathcal S}^{(n,k)} = [0,1)$.
\end{theorem}

It follows from Theorem \ref{thm:maximal-stability} that $\rho$ in \eqref{rho} is indeed the traffic intensity, as there exist policies that stabilize the system for every $\rho < 1$, but no such policy exists when $\rho \ge 1$. %

\section{The Induced $(k,k)$ System} \label{section:induce}

A difficulty in the analysis of $(n,k)$ systems is that a single service completion may trigger the simultaneous removal of multiple redundant tasks, while the identity of the tasks to be removed is not known in advance. To address this challenge, we project the $(n,k)$ system onto a lower-dimensional space and interpret the resulting projected process as the queue process of a $(k,k)$ system (which has no redundancy). We refer to this projected system as the \emph{induced} $(k,k)$ system. We then show that the induced and inducing systems are either both stable or both unstable (see Lemma \ref{LemStable} below). Consequently, stability of the original FJR system can be established by analyzing the simpler induced $(k,k)$ system.

\subsection{The Projection Mapping} \label{subsection:inducedSystem}

For $i \in \llbracket 1, k\rrbracket$ and $\bm q \in \mathbb S^{(n,k)}$, let $\phi_i(\bm q) := q_{\sigma_i}$ and define the map $\bm \phi:\mathbb S^{(n,k)} \to \mathcal R(\Z_+^k)$  via
\begin{equation*}
\bm \phi(\bm q) = \bigl(\phi_1(\bm q), \ldots, \phi_k(\bm q) \bigr) = (q_{\sigma_1}, \ldots, q_{\sigma_k}).
\end{equation*}
Then $\bm \phi$ is an operator that maps each $\bm q \in \mathbb S^{(n,k)}$ to the $k$-dimensional vector of its $k$ smallest components, arranged in ascending order.

Next, for a given policy $\pi\in\Pi^{(n,k)}$, define the map $\bm \psi^\pi: \mathbb S^{(n,k)} \rightarrow \Delta^k$ via
\begin{equation*}
\psi_i^\pi(\bm q):=\gamma_{\sigma_i}^\pi(\bm q),\quad i\in \llbracket 1,k-1 \rrbracket,
\qquad
\psi_k^\pi(\bm q):=\sum_{j=k}^n \gamma_{\sigma_j}^\pi(\bm q).
\end{equation*}
Specifically, for each $i\in \llbracket 1,k-1 \rrbracket$, the $i$th coordinate of $\bm \psi^\pi(\bm q)$ corresponds to the service capacity of the $i$th shortest queue when the system is in state $\bm q$. The final coordinate $\psi^{\pi}_k(\bm q)$ is the total service capacity allocated to the remaining $n-k+1$ (longest) queues.
We refer to $\bm x = \bm \phi(\bm q)$ as the \textit{induced state}, and to $\bm y = \bm \psi^\pi(\bm q)$ as the \textit{induced action}. Then the induced system has a direct transition from $\bm x = \bm \phi(\bm q)$ to $\bm x' = \bm \phi( \bm q')$ if and only if $\bm q \to \bm q'$, and corresponding transitions in both systems occur at the same rate. We thus use the transition notation $\bm x \to \bm x'$ in this case.

\begin{theorem}\label{thm:induce}
Consider an $(n,k)$ system operating under policy $\pi\in\Pi^{(n,k)}$. For $\bm q\in \mathbb S^{(n,k)}$, let $\bm x = \bm \phi (\bm q)$ and $\bm y = \bm \psi^\pi (\bm q)$. %
Then 
\begin{align}
&\bm x\to \bm x+\bm e  \quad \text{at rate }\lambda;\label{eq:kk_transition1}\\
&\bm x\to \mathcal R\bigl((\bm x-\bm e_i)^+\bigr) \quad \text{at rate } y_i,\quad i\in \llbracket 1, k \rrbracket.\label{eq:kk_transition2}
\end{align}
\end{theorem}
The proof of Theorem \ref{thm:induce} appears in \S \ref{app:induce}. 

Theorem \ref{thm:induce} demonstrates that the induced system can be viewed as a $(k,k)$ system with its queues numbered in nondecreasing order (namely, with the first queue being the smallest).  
Indeed, consider a $(k,k)$ system in which the $i$th shortest queue, having length $x_i$, is served at rate $y_i$, $i\in \llbracket 1, k \rrbracket$ at some time $t$. An arrival of a job, which occurs at rate $\lambda$, increases each queue by $1$ and does not change the order of queue lengths. Therefore, $\bm x \to \bm x + \bm e$ at rate $\lambda$. A task completion at the $i$th shortest queue, which happens at rate $y_i$, will decrease $x_i$ by $1$ if $x_i>0$, followed by reordering the queue vector. Therefore, the resulting state is $\mathcal R\big((\bm x - \bm e_i)^+\big)$, and we have $\bm x \to \mathcal R\big((\bm x - \bm e_i)^+\big)$ at rate $y_i$. In particular, the transitions of this  $(k,k)$ system with ordered queue lengths are the same as the transitions of the induced system in \eqref{eq:kk_transition1} and \eqref{eq:kk_transition2}.

Then $(\bm X^\pi, \bm Y^\pi)$, defined via $\bm X^\pi(t) := \bm \phi \bigl( \bm Q^\pi(t) \bigr)$ and $\bm Y^\pi(t) := \bm \psi^\pi \bigl( \bm Q^\pi(t) \bigr)$, can be interpreted as the ordered queue process and the ordered capacity-allocation process, respectively, of a $(k,k)$ system in which the $i$th shortest queue at time $t$ (having length $X^\pi_i(t)$) is processed at a rate $Y^\pi_i(t)$, $i\in \llbracket 1, k \rrbracket$.
Note that $(\bm X^\pi, \bm Y^\pi)$ is not necessarily Markov, since its dynamics are governed by the transitions of the original $(n,k)$ system. We thus say that the induced $(k,k)$ system is stable if it regenerates in finite expected time, with its regeneration state being the empty state. 

Observing that
\begin{equation} \label{eq:identity-z-x}
Z^\pi(t) = \max_{i\in \llbracket 1, n \rrbracket} Q^\pi_i(t) = \max_{i\in \llbracket 1, k \rrbracket} X^\pi_i(t) = X^\pi_k(t),
\end{equation}
where the first equality is due to \eqref{eq:maxq}, the second equality follows from the definition of $\bm \phi$, and the last equality holds because $\bm X^\pi(t)\in \mathcal R(\Z_+^k)$, the following result is immediate.

\begin{lemma} \label{LemStable}
    An $(n,k)$ system is stable if and only if its induced $(k,k)$ system is stable.
\end{lemma}

\subsection{The GSC Order} \label{subsection:kkproperty}

The proofs of the main results build on the following GSC order for vectors in $\R^k$; see \cite[\S 4]{moyal2022stability}) for background.

\begin{definition}[GSC order]
	For $\bm a, \bm b\in \R^k$, we write $\bm a \le_{\text{GSC}} \bm b$ if
\begin{equation*}
	\sum_{i=m}^k a_i \le \sum_{i=m}^k b_i \text{ for all } m \in \llbracket 1, k\rrbracket.
\end{equation*}
\end{definition}
 Based on the GSC order for vectors, we consider the following stochastic order for processes.

\begin{definition}[GSC sample-path stochastic order] \label{def:GSCstochstic}
	For two $\R^k$-valued stochastic processes $\bm A$ and $\bm B$, we write $\bm A \le_{\text{st,GSC}} \bm B$ if there exist two stochastic processes $\mathcal{\bm{A}}$ and $\mathcal{\bm{B}}$ defined on a common probability space, such that $\bm A \deq \mathcal{\bm A}$, $\bm B \deq \mathcal{\bm B}$, and $\mathcal{A}(t) \le_{\text{GSC}} \mathcal{B}(t)$ w.p.1 for all $t \ge 0$. 
\end{definition}

\begin{lemma} \label{lem:stochasticOrder}
	Consider two $(n,k)$ systems operating under policies $\pi$ and $\pi^{\prime}$, both having the same arrival rate $\lambda$, and consider the respective induced $(k,k)$ systems $(\bm X^\pi, \bm Y^\pi)$ and $(\bm X^{\pi^\prime}, \bm Y^{\pi^\prime})$, where $\bm Y^{\pi^\prime} = \bm b \eta$, for some $\bm b \in \Delta^k$.
    \begin{enumerate}
    \item[(i)] If $\bm Y^\pi \ge_{\text{st,GSC}} \bm Y^{\pi^\prime}$, then  $\mathcal S^{(n,k)} (\pi^\prime) \subseteq \mathcal S^{(n,k)} (\pi)$.

    \item[(ii)] If $\bm Y^{\pi^\prime} \ge_{\text{st,GSC}}~ \bm Y^{\pi}$, then  $\mathcal S^{(n,k)} (\pi) \subseteq \mathcal S^{(n,k)} (\pi^\prime)$.
    \end{enumerate}
\end{lemma}     
The proof of Lemma \ref{lem:stochasticOrder} appears in \S \ref{app:order}.

Lemma \ref{lem:stochasticOrder} is the key to our stability analysis. It allows us to establish the stability of an $(n,k)$ system by bounding, in the GSC order, the sample paths of its induced $(k,k)$ system by those of a simpler $(k,k)$ system that operates under a constant capacity-allocation process. The stability of the original $(n,k)$ system then follows from Lemma \ref{LemStable}. %

\subsection{Proof of Theorem \ref{thm:maximal-stability}} \label{subsection:maximalStability}

We are now ready to prove Theorem \ref{thm:maximal-stability}. In applying Lemma \ref{lem:stochasticOrder} we will compare the sample path of the induced $(k,k)$ system to that of a $(k,k)$ system operating under the control that gives all the service capacity to the longest queue at any time, denoted by $\pi_{\text{pool}}$, so that
$
\bm Y^{\pi_{\text{pool}}} = \bm e_k \eta.
$

\begin{lemma}\label{lemma:kk}
	$\mathcal S^{(k,k)}(\pi_{\text{pool}}) = [0,1)$. 
\end{lemma}

\proof{Proof.}
Since the arrival of a job does not affect the ordering of the queue lengths, the total capacity is dynamically allocated to each of the stations on a round-robin basis; in particular, the oldest job in the system departs at the end of each cycle. It follows that the service time of each job is the sum of $k$ i.i.d.\ unit-rate exponential service times. Hence, the system evolves as an $M/G/1$ queue with arrival rate $\lambda$ and Erlang-distributed service times with shape $k$ and rate $1$, which is stable if and only if $\lambda < 1/k$, or equivalently, if and only if $\rho = k \lambda < 1$. \hfill \Halmos

\endproof

The next lemma establishes that $\rho < 1$ is a necessary condition for stability of $(k,k)$ systems. 
\begin{lemma} \label{lemma:maximalkk}
	$\bar{\mathcal S}^{(k,k)} = [0,1)$.
\end{lemma}

\proof{Proof.}
By Lemma \ref{lemma:kk}, there exists an admissible policy for which $[0,1)$ is the stability region. On the other hand, 
$$\bm y \le_{\text{GSC}} \bm e_k, \quad \text{for any } \bm y \in \Delta^k.$$ In turn, a $(k,k)$ system under any admissible policy $\pi$ must have an ordered capacity-allocation process $\bm Y^\pi$ with  
$\bm Y^\pi \le_{\text{st,GSC}} \bm e_k \eta$.
 By taking $n=k$, $\bm b = \bm e_k$, and $\pi^\prime = \pi_{\text{pool}}$ in Lemma \ref{lem:stochasticOrder} (when applied to the special case of $n=k$), we have $\mathcal S^{(k,k)} (\pi) \subseteq \mathcal S^{(k,k)} (\pi_{\text{pool}})$. Hence, $\bar{\mathcal S}^{(k,k)} = [0,1)$.  \hfill \Halmos

\endproof

\proof{Proof of Theorem \ref{thm:maximal-stability}.}

Consider an $(n,k)$ system under an arbitrary policy $\pi \in \Pi^{(n,k)}$, with induced processes $\bm X^\pi = \bm \phi(\bm Q^\pi)$ and $\bm Y^\pi = \bm \psi^\pi (\bm Q^\pi)$. 
By Lemma \ref{lemma:maximalkk}, the induced system is not stable if $\rho \ge 1$, and therefore the original $(n,k)$ system is unstable as well by virtue of Lemma \ref{LemStable}. Hence $\bar{\mathcal S}^{(n,k)} \subseteq [0,1)$.

Conversely, for any policy $\pi_k \in \Pi^{(k,k)}$ applied to a $(k,k)$ system, there exists a policy $\pi_n \in \Pi^{(n,k)}$ under which the two systems are equal in distribution. In particular, for such a $\pi_k$ we take $\pi_n$ to be the policy that allocates no capacity to stations $k+1, \ldots, n$, and uses the same allocation rule as $\pi_k$ to allocate capacity to stations $1,\dots, k$. It follows that the maximal stability region of $(k,k)$ systems is no larger than that of $(n,k)$ systems. Hence, $\bar{\mathcal S}^{(k,k)} \subseteq \bar{\mathcal S}^{(n,k)}$, so that $[0,1) \subseteq \bar{\mathcal S}^{(n,k)}$ by Lemma \ref{LemStable}.

We conclude that $\bar{\mathcal S}^{(n,k)} = [0,1)$, as stated. \hfill \Halmos

\endproof

\section{Maximal Static Policies} \label{section:static}

We now consider the maximal-design problem for $(n,k)$ systems under static policies in $\hat{\Pi}^{(n,k)}$. Specifically, we show that the maximal stability region can be achieved by a static policy, and characterize maximal static policies.
To this end, recall that for $\bm d\in \Delta^n$, the static policy $\pi_{\bm d}$ satisfies $\bm \gamma^{\pi_{\bm d}} (\bm q) \equiv \bm d$ for every $\bm q\in \mathbb S^{(n,k)}$; that is, the server at station $i$ processes work at the constant rate $d_i$ throughout. Let
\begin{equation*} 
    d_{\min} := \min_{i\in \llbracket 1, n \rrbracket } d_i.
\end{equation*}
The next lemma characterizes the stability region of a $(k,k)$ system operating under the static policy $\pi_{\bm d}$.
\begin{lemma} \label{lemma:kkstatic}
		$\mathcal S^{(k,k)} (\pi_{\bm d}) = [0, k d_{\min} )$. %
\end{lemma} 

\proof{Proof.}
Under $\pi_{\bm d}$, a $(k,k)$ system is simply $k$  (nonhomogeneous) parallel $M/M/1$ queues fed by the same Poisson arrival process. Hence, the system is stable if and only if each of these $M/M/1$ queues is stable, namely, if and only if $\lambda < d_i$ for all $i\in \llbracket 1, k \rrbracket$, or equivalently, if and only if $\rho = k\la < k d_{\min}$.\hfill \Halmos
\endproof

Next, consider the homogeneous capacity-allocation vector 
$
	\bm{d}^k := (1/k, 1/k, \ldots, 1/k) \in \Delta^k.
$
Since $d^k_{\min} = 1/k$, Lemma \ref{lemma:kkstatic} yields
	$\mathcal S^{(k,k)} (\pi_{\bm{d}^k}) = [0,1)$.
Thus $\pi_{\bm{d}^k}$ is a maximal static policy for the $(k,k)$ system. The same dynamics are obtained in an $(n,k)$ system under the static policy $\pi_{\tilde{\bm{d}}^k} \in \hat{\Pi}^{(n,k)}$, where $\tilde{\bm{d}}^k \in \Delta^n$ has components
\begin{equation} \label{eq:d^k}
	 \tilde{d}^k_i = 1/k,  \qforq i\in \llbracket 1, k \rrbracket  \quad \text{and} \quad \tilde{d}^k_i = 0, \qforq i\in \llbracket k+1, n \rrbracket.
\end{equation}

\begin{corollary}
	$\mathcal S^{(n,k)} (\pi_{\tilde{\bm{d}}^k} ) = [0,1)$. Thus, $\pi_{\tilde{\bm{d}}^k}$ is a maximal static policy. 
\end{corollary}
We conclude that the maximal stability region $[0,1)$ is achievable by at least one static policy in $(n,k)$ systems. 

\paragraph{\textbf{Non-robustness of maximality without redundancy.}}
Observe that in a $(k,k)$ system, $\pi_{\bm d^k}$ is the unique maximal static policy. Indeed, for any $\bm d \in \Delta^k$ with $\bm d\neq \bm d^k$, we have $k d_{\min} < 1$, so the stability region is strictly smaller than $[0,1)$. Similarly, if we consider static policies $\pi_{\bm d}$ of the form $d_{k+1} = d_{k+2} = \cdots = d_n = 0$ in an $(n,k)$ system (under which the $(n,k)$ system is effectively a $(k,k)$ system), any perturbation in the component values of $\tilde{\bm d}^k$ given in \eqref{eq:d^k} strictly reduces the stability region.
We next show that static policies are substantially more robust to perturbations in the service capacities when redundancy is introduced. 

\subsection{Maximal-Design Criterion} \label{subsection:sufficient}

Given $\bm d\in \Delta^n$ and a static policy $\pi_{\bm d} \in \hat \Pi^{(n,k)}$, the induced $(k,k)$ system has the ordered capacity-allocation process $\bm Y^{\pi_{\bm d}} = \bm \psi^{\pi_{\bm d}} (\bm Q^{\pi_{\bm d}})$. %
Note that the capacity-allocation process $\bm \Gamma^{\pi_{\bm d}}$ in the original $(n,k)$ system is constant under a static policy. By contrast,  $\bm Y^{\pi_{\bm d}}$ need not be a constant process, because the index $\sigma_i(\bm Q(t))$ of the $i$th shortest queue depends on the state, and therefore keeps changing. In particular, the policy in the induced $(k,k)$ system of an $(n,k)$ system operating under a static policy is not static.

For $\bm d\in \Delta^n$, let $d_M := \max_{i\in \llbracket 1, n \rrbracket } d_i$. %

\begin{theorem}[Criterion for maximal design] \label{thm:sufficient}
If $d_M \le 1/k$, then $\pi_{\bm d}$ is maximal.
\end{theorem}

\proof{Proof.}

For $\bm d\in \Delta^n$ with $d_M \le 1/k$, the induced process $\bm Y^{\pi_{\bm d}} = \bm \psi^{\pi_{\bm d}} (\bm Q^{\pi_{\bm d}})$ satisfies 
$\bm Y^{\pi_{\bm d}} \ge_{\text{st,GSC}} \bm d^k \eta$,
where $\bm d^k = (1/k, \ldots, 1/k)$.
Indeed, 
$$
Y_i^{\pi_{\bm d}}(t) = d_{\sigma_i\left(\bm Q^{\pi_{\bm d}}(t) \right)} \le d_M \le 1/k \quad \text{w.p.1, for all } i\in \llbracket 1, k-1 \rrbracket \text{ and } t \ge 0.
$$ 
Then $\mathcal S^{(n,k)} (\pi^{\prime}) \subseteq \mathcal S^{(n,k)} (\pi_{\bm d})$  by Lemma \ref{lem:stochasticOrder}, where the $(n,k)$ system under policy $\pi^{\prime}$ has the induced process $\bm Y^{\pi^\prime} = \bm d^k \eta$, i.e., its induced $(k,k)$ system has homogeneous servers, each processing work at rate $1/k$. By Lemma \ref{lemma:kkstatic}, the stability region of the latter system is $[0,1)$.
We thus have $\mathcal S^{(n,k)} (\pi^{\prime}) = [0,1) \subseteq \mathcal S^{(n,k)} (\pi_{\bm d})$, and the statement follows because $\mathcal S^{(n,k)} (\pi_{\bm d}) \subseteq \bar{\mathcal S}^{(n,k)} = [0,1)$. \hfill \Halmos
\endproof

The condition $d_M \le 1/k$, which we refer to as the \emph{maximal-design criterion}, guarantees that the $k-1$ shortest queues at any time $t \ge 0$ are each served at a rate no greater than $1/k$. Hence, the total service capacity allocated to the remaining $n-k+1$ queues is at least $1/k$. It follows that, in the induced $(k,k)$ system, each of the first $k-1$ queues is served at a rate no greater than $1/k$, while the longest queue is served at a rate no less than $1/k$. (Recall that the induced system does not operate under a static policy.)

\paragraph{\textbf{Robustness due to redundancy.}}
The value $d_M$ can be viewed as a measure of server heterogeneity. Its minimum value is $1/n$, corresponding to homogeneous servers, while its maximum value is $1$, in which case $n-1$ stations receive zero service capacity and the system is unstable for every arrival rate $\lambda>0$. Theorem \ref{thm:sufficient} shows that the maximal stability region $[0,1)$ is robust to a certain level of heterogeneity: a static policy is maximal as long as no server is allocated more than $1/k$ of the total service capacity, regardless of how that capacity is distributed among the servers. In particular, maximality holds whenever $d_M\in[1/n,1/k]$.

In a $(k,k)$ system without redundancy, the interval $[1/n,1/k]$ collapses to the singleton $\{1/k\}$, so that any perturbation of the service-capacity allocation strictly reduces the stability region. In many practical settings, however, such perturbations are unavoidable. For example, in server farms used for multi-agent LLM inference, imbalanced power allocation, inaccuracies in the power-to-frequency relationship, and fluctuations in processor speeds due to environmental conditions may all lead to deviations from an ideal allocation. The benefit of increased redundancy (i.e., increasing the number of servers $n$) is that it enlarges the interval $[1/n,1/k]$, thereby allowing a broader class of capacity allocations $\bm d \in \Delta^n$ to achieve the maximal stability region.

\section{Maximal Dynamic Policies} \label{section: dynamic}

We now consider the class of dynamic policies $\Pi^{(n,k)}$. In this setting, establishing the stability condition for a given policy can be difficult, because the queue process is multidimensional and the long-run job-completion rate is difficult to characterize, as the service capacity allocated to each server changes continuously over time.
Thus, our goal is to characterize a general \textit{maximal-control criterion} which guarantees that a dynamic policy is maximal. 

Recall that for $\bm q \in \mathbb S^{(n,k)}$, $\sigma_i := \sigma_i(\bm q)$ denotes the index of the $i$th shortest queue, $i\in \llbracket 1, n \rrbracket$.

\begin{theorem}[Criterion for maximal control] \label{thm:maximalControl}
	Consider an $(n,k)$ system operating under a policy $\pi \in \Pi^{(n,k)}$. If
	\begin{equation} \label{eq:maximalcontrol}
	 \sum_{i=1}^j	\gamma^\pi_{\sigma_i} (\bm q) \le \frac{j}{k} \quad \text{ for all } j\in \llbracket 1, k-1 \rrbracket \text{ and all } \bm q\in \mathbb S^{(n,k)},
	\end{equation}
	then
$
		\mathcal S^{(n,k)} (\pi) = [0,1).
$
	In particular, $\pi$ is maximal.
\end{theorem}

\proof{Proof.}
The process $\bm Y^\pi = \bm \psi^\pi(\bm Q^\pi)$ in the induced system satisfies
\begin{equation*}
	\sum_{i=1}^j	Y_i^\pi(t) = \sum_{i=1}^j \gamma^\pi_{\sigma_i} (\bm Q^\pi(t)) \le \frac{j}{k}, \quad \text{for any } j\in \llbracket 1, k-1 \rrbracket \text{ and } t\ge 0.
\end{equation*}
Since $\bm Y^\pi$ takes values in $\Delta^k$, the above is equivalent to
\begin{equation*}
	\sum_{i=j+1}^k Y_i^\pi(t) \ge \frac{k-j}{k}, \quad \text{for any } j \in \llbracket 1, k-1 \rrbracket \text{ and any } t\ge0,
\end{equation*}
so that
\begin{equation*}
\bm Y^\pi \ge_{\text{st,GSC}} \bm d^k \eta, \text{ for }\bm d^k = (1/k, 1/k, \ldots, 1/k)\in \Delta^k.
\end{equation*}
From here, similar arguments to those in the proof of Theorem \ref{thm:sufficient} give that
$\mathcal S^{(n,k)} (\pi) = [0,1)$. \hfill \Halmos 
\endproof

The criterion for maximality of dynamic policies is weaker than the criterion for maximal design. Indeed, the maximal design criterion $d_M\le 1/k$ implies that the capacity-allocation process $\bm \Gamma$ satisfies
\begin{equation} \label{eq:staticProcess}
	\Gamma_j(t) \le \frac{1}{k}, \text{ for each } j\in \llbracket 1, n \rrbracket \text{ and any }t\ge 0.
\end{equation}
The maximal-control criterion given in Theorem \ref{thm:maximalControl} requires only that the cumulative capacity allocated to the first $j$ shortest queues is no more than $j/k$ for each $j\in \llbracket 1, k-1 \rrbracket$. In particular, the allocation to the remaining $n-(k-1)$ queues can be arbitrary. Thus, the capacity-allocation process satisfies
 \begin{equation} \label{eq:dynamicProcess}
 	\sum_{i=1}^j \Gamma_{\sigma_i}(t) \le \frac{j}{k}, \text{ for each } j\in \llbracket 1, k-1 \rrbracket \text{ and any }t\ge 0. %
 \end{equation}
 Clearly, \eqref{eq:dynamicProcess} holds if \eqref{eq:staticProcess} holds, but not vice versa.

\paragraph{\textbf{Implications for systems with flexible and collaborative servers.}}
Consider the $(n,k)$-FC system in Example \ref{exFlexibleServers}, with $m \ge 1$ servers having service rates $\bm \mu = (\mu_1, \dots, \mu_m)$ such that $\sum_{j=1}^m \mu_j = 1$. In this setting, the total service capacity consists of the $m$ components of $\bm \mu$ and therefore cannot be divided arbitrarily, with the capacity-allocation process $\bm \Gamma$ taking values in $\Delta^n_{\bm \mu}$ defined in \eqref{eqDelta_u}. If $\mu_r > 1/k$ for some $r\in \llbracket 1,m\rrbracket$, then the maximal-design criterion $d_M \le 1/k$ cannot hold for any $\bm d\in \Delta^n_{\bm \mu}$. In contrast, a dynamic policy can allocate the components of $\bm \mu$ so that the maximal-control criterion for dynamic policies holds, by assigning no service capacity to one or more of the $k-1$ shortest queues whenever necessary. Thus, unlike static policies, there exist dynamic policies that are guaranteed to be maximal for any $\bm \mu$.

\section{Summary and Future Research}\label{section:summary}

We studied the stability problem of $(n,k)$ systems under both static and dynamic capacity-allocation policies. In particular, we identified the nominal traffic intensity, characterized the maximal stability region, and established conditions under which both static and dynamic policies attain that region.

We first showed that the maximal stability region (among all policies) is attainable by static policies. We then characterized a maximal-design criterion, under which the static policy is guaranteed to be maximal. Specifically, we proved that if the largest allocated capacity $d_M$ among all servers satisfies $d_M \le 1/k$, then such a static policy is maximal. This criterion shows that redundancy makes maximality more robust to the heterogeneity of service capacities.

For dynamic policies, we derived a maximal-control criterion requiring that, at every state, the cumulative capacity allocated to the servers with the $j$ shortest queues be no more than $j/k$ for each $j=1,\ldots,k-1$; the remaining capacity allocation for the remaining $n-k+1$ servers can be arbitrary. Compared to the maximal design criterion for static policies, the criterion for maximal dynamic policies is substantially relaxed, so a much broader set of policies can be maximal.

A key challenge in the analysis of an $(n,k)$ system is the complexity of its state space stemming from the redundancy mechanism. To overcome this difficulty, we projected the $(n,k)$ system onto a lower-dimensional $(k,k)$ system, which we termed the induced system. We then established a sample-path comparison result based on a generalized Schur-convex (GSC) order for coupled $(k,k)$ systems. The resulting GSC sample-path stochastic-order comparisons enabled us to establish the stability results for the induced system and, in turn, for the original $(n,k)$ system. 

\paragraph{Significance of Our Results for Future Research.}
Given our characterization of the maximal policies, a natural direction for future research is to determine which maximal policies outperform others with respect to fundamental performance measures, such as the steady-state mean sojourn time and mean waiting time. In particular, our work here is a necessary first step toward the analysis, design, and control of FJR systems, because such performance measures should be optimized only among maximal policies. The reason is that these measures grow highly nonlinearly as the traffic intensity approaches the boundary of the stability region under a given policy. Consequently, an arrival rate that places a system in heavy traffic, or even renders it unstable, under a non-maximal policy may correspond to a system operating well within its stability region under a maximal policy. Optimizing performance measures over all admissible policies therefore risks selecting policies that are fundamentally constrained by an unnecessarily small stability region.

\section{Remaining Proofs}\label{app:proofs}

This section contains the remaining proofs, in addition to auxiliary results to support those proofs.

\subsection{Proof of Proposition~\ref{prop:q-transitions}}\label{app:q-transitions}

\proof{Proof.}
The transition due to an arrival in \eqref{eq:q-arrival} is immediate. 

The transition in \eqref{eq:q-departure-nojoin} follows because, when $\|\bm v\|_0 < k-1$, each of the jobs in the system has less than $k-1$ completed tasks. Therefore, a completion of a task at queue $i$ decreases $q_i$ by $1$, without triggering the removal of its sibling tasks..

Now consider the case $\|\bm v\|_0=k-1$.
For station $i$ with $q_i>0$ at time $t$, the join-queue has length $v_i$ and contains tasks $\{\tau_{1,i},\ldots,\tau_{v_i,i} \}$ with $\tau_{v_i+1,i}$ being the task in service. The sibling task $\tau_{v_i+1, \ell}$ of $\tau_{ v_i+1 ,i}$ at station $\ell$ %
is in join-queue $\ell$ if and only if $v_\ell \ge v_i+1$. Hence, the number of completed sibling tasks of task $\tau_{v_i+1, i}$ is
\begin{equation*} 
  c_i := \sum_{\ell\ne i} \mathbbm{1} \{v_\ell \ge v_i+1 \}.
\end{equation*}
Note that $c_i \le k-1$, and that the completion of task $\tau_{v_i+1, i}$  triggers a job departure and removal of its sibling tasks if and only if $c_i = k-1$.
It follows from \eqref{eq:join-empty-largest} and the equality $\|\bm v\|_0=k-1$,  that
\begin{equation} \label{eq:order_v}
	v_{\sigma_1} \ge \cdots \ge  v_{\sigma_{k-1}} > 0 = v_{\sigma_k}=\cdots=v_{\sigma_n}.
\end{equation}

Consider first a task completion at station $\sigma_i$ with $i\le k-1$. %
\begin{equation*}
	c_{\sigma_i} = \sum_{\ell\ne \sigma_i} \mathbbm{1} \{v_\ell \ge v_{\sigma_i}+1 \}  = \sum_{p < i} \mathbbm{1} \{v_{\sigma_p} \ge v_{\sigma_i}+1 \} \le \sum_{p < i} \mathbbm{1} \{v_{\sigma_p} \ge v_{\sigma_i} \} = i-1\le k-2.
\end{equation*}
Hence, a completion of a task at station $\sigma_i$ does not trigger a job departure, and the process transitions to $(\bm q-\bm e_{\sigma_i})^+$ at rate $\gamma^\pi_{\sigma_i} (\bm q)$, proving \eqref{eq:q-departure-short}. %

Finally, to prove \eqref{eq:q-departure-long}, consider the completion of a task at station $\sigma_r$ with $r\ge k$. %
By \eqref{eq:order_v}, we have $v_{\sigma_r} = 0$, and the task in service is $\tau_{v_{\sigma_r} + 1,\sigma_r} = \tau_{1,\sigma_r}$. The number of its completed sibling tasks is
\begin{equation*}
	c_{\sigma_r} = \sum_{\ell\ne \sigma_r} \mathbbm{1} \{v_\ell \ge v_{\sigma_r}+1 \} = \sum_{\ell\ne \sigma_r} \mathbbm{1} \{v_\ell \ge 1 \} = \sum_{\ell = 1}^n \mathbbm{1} \{v_\ell \ge 1 \} = \|\bm v\|_0 = k-1.
\end{equation*}
Then the completion of task $\tau_{1,\sigma_r}$ leads to the departure of job $1$, so that all this job's tasks $\{\tau_{1,\sigma_1}, \ldots, \tau_{1,\sigma_n}\}$ are removed from the system. Specifically, by \eqref{eq:order_v}, tasks $\{\tau_{1,\sigma_1},\ldots,\tau_{1,\sigma_{k-1}} \}$ are already in the join-queues of stations $\sigma_1, \ldots, \sigma_{k-1}$; tasks $\{\tau_{1,\sigma_k},\ldots,\tau_{1,\sigma_{n}} \}$ are in the queues of stations $\sigma_k, \ldots, \sigma_{n}$. Hence, the removal of the tasks of job $1$ decreases every queue length at stations $\sigma_k, \ldots, \sigma_{n}$ by one. It follows that the process transitions to state 
$
\left(\bm q-\sum_{j=k}^n \bm e_{\sigma_j}\right)^+
$
at rate $\gamma^\pi_{\sigma_r} (\bm q)$ for each $r\in \llbracket k, n \rrbracket$, and therefore at a total rate of $\sum_{r=k}^n \gamma^\pi_{\sigma_r} (\bm q)$. %
\hfill \Halmos
\endproof

\subsection{Proof of Theorem \ref{thm:induce}} \label{app:induce}

In the proof of Theorem \ref{thm:induce} we will use the following lemma, whose proof appears in \S \ref{app:projection-identities}. 

\begin{lemma}\label{lem:projection-identities}
The following hold for $\bm q\in \mathbb S^{(n,k)}$ and $\bm x := \bm \phi (\bm q)$.
\begin{align}
\bm \phi \bigl((\bm q-\bm e_{\sigma_r})^+\bigr)
&=
\mathcal R\bigl(( \bm x-\bm e_r)^+\bigr),
\qquad r\in \llbracket 1, k-1 \rrbracket, \label{lem:projection-short}\\
\bm \phi\bigl((\bm q - \bm e_{\sigma_j})^+\bigr)
&=
\mathcal R\bigl(( \bm x -\bm e_k)^+\bigr),
\qquad j\in \llbracket k, n \rrbracket, \label{lem:projection-long-single}
\end{align}
Moreover, if $\|\bm v (\bm q)\|_0 = k-1$, then
\begin{equation}\label{lem:projection-long-block}
\bm \phi\left(\left(\bm q-\sum_{j=k}^n \bm e_{\sigma_j}\right)^+\right)
=
\mathcal R\bigl(( \bm x - \bm e_k)^+\bigr).
\end{equation}
\end{lemma}

\proof{Proof of Theorem \ref{thm:induce}.}
 By Proposition~\ref{prop:q-transitions}, a transition of the queue process $\bm Q$ from a state $\bm q$ is to one of the states on the right-hand sides of the arrows in \eqref{eq:q-arrival}--\eqref{eq:q-departure-long}  (the ``post-transition states''). %

For the transition \eqref{eq:q-arrival}, the induced pre-transition state is $\bm \phi(\bm q) = \bm x$, and since $\bm q+\bm e$ does not change the order of the components, the induced post-transition state is $\bm \phi( \bm q+ \bm e )  = \bm x + \bm e$. Hence,
\begin{equation} \label{eq:inducedArrival}
\bm x \rightarrow \bm x + \bm e \quad \text{at rate } \lambda.
\end{equation}

Next, consider the transitions corresponding to departures.
When $\|\bm v\|_0 < k-1$, the induced transitions of \eqref{eq:q-departure-nojoin} are
\begin{align*}
	& \bm x \rightarrow \bm \phi\bigl( (\bm q - \bm e_{\sigma_i})^+ \bigr)  = \mathcal R\bigl( (\bm x - \bm e_{i})^+ \bigr) \quad \quad \text{at rate } \gamma_{\sigma_i}^\pi(\bm q) = y_i ~\text{ for } i \in \llbracket 1, k-1 \rrbracket; \\
	& \bm x \rightarrow \bm \phi\bigl( (\bm q - \bm e_{\sigma_i})^+ \bigr)  = \mathcal R\bigl( (\bm x - \bm e_{k})^+ \bigr) \quad \quad \text{at rate } \gamma_{\sigma_i}^\pi(\bm q)  ~\text{ for } i \in \llbracket k, n \rrbracket,
\end{align*}
where the equalities follow from \eqref{lem:projection-short} and \eqref{lem:projection-long-single} in Lemma \ref{lem:projection-identities}. Note that the induced pre-transition and post-transition states are identical for all $i \in \llbracket k, n \rrbracket$, so the transition rates corresponding to each $i$ can be aggregated into $y_k := \sum_{i=k}^n \gamma^\pi_{\sigma_i}(\bm q)$, and therefore
\begin{equation} \label{eq:inducedDeparture1}
	\bm x \rightarrow  \mathcal R\bigl( (\bm x - \bm e_{i})^+ \bigr) \quad  \text{at rate } y_i ~\text{ for } i \in \llbracket 1, k \rrbracket. 
\end{equation}

Finally, when $\|\bm v\|_0 = k-1$, the induced transitions of \eqref{eq:q-departure-short} and \eqref{eq:q-departure-long} are, respectively,
\begin{align*}
	& \bm x\to \bm \phi\bigl( (\bm q-\bm e_{\sigma_i})^+ \bigr) = \mathcal{R} \bigl( (\bm x - \bm e_i)^+ \bigr) \quad \quad \quad ~~ \quad \text{at rate }\gamma_{\sigma_i}^\pi(\bm q) = y_i,\quad i\in \llbracket 1, k-1 \rrbracket;\\
	&\bm x \to \bm \phi \left(  \left(\bm q-\sum_{j=k}^n \bm e_{\sigma_j} \right)^+ \right) = \mathcal R\bigl(( \bm x - \bm e_k)^+\bigr)   \quad \text{at rate } \sum_{j=k}^n \gamma_{\sigma_j}^\pi(\bm q) = y_k,
\end{align*} 
where the equalities of the post-transition states follow from \eqref{lem:projection-short} and \eqref{lem:projection-long-block}. Equivalently, 
\begin{equation} \label{eq:inducedDeparture2}
	\bm x \rightarrow  \mathcal R\bigl( (\bm x - \bm e_{i})^+ \bigr) \quad \text{at rate } y_i ~\text{ for } i \in \llbracket 1, k \rrbracket. 
\end{equation}

It follows from \eqref{eq:inducedDeparture1} and \eqref{eq:inducedDeparture2} that there is no difference between the two cases $\|\bm v\|_0 < k-1$ and $\|\bm v\|_0 = k-1$, and thus the statement of the theorem follows from \eqref{eq:inducedArrival}-\eqref{eq:inducedDeparture2}. 
\hfill \Halmos \endproof

\subsection{Proof of Lemma~\ref{lem:projection-identities}}\label{app:projection-identities}

\proof{Proof.}
Recall that $\bm \phi(\bm q)$ is obtained by taking the $k$ smallest components of $\bm q \in \Z^n_+$ and reordering them in nondecreasing order, and that $\mathcal R(\bm y)$ orders the components of a vector $\bm y \in \Z^k_+$ in nondecreasing order. %
Hence, \eqref{lem:projection-short} is immediate.

To show \eqref{lem:projection-long-single}, recall that, by \eqref{StateSpace}, $\bm q \in \mathbb S^{(n,k)}$ implies that $q_{\sigma_k} =\cdots=q_{\sigma_n}$. For any fixed $j\in \llbracket k, n \rrbracket$, the components of $(\bm q-\bm e_{\sigma_j})^+$ are then equal to
\begin{equation*}
\{x_1,\ldots,x_{k-1},\,(x_k-1)^+,\,x_k,\ldots,x_k\}, \qforq \bm x = \bm \phi(\bm q),
\end{equation*}
with $n-k$ components equal to $x_k$. Since
\begin{equation*}
x_1 \le\cdots\le x_{k-1} \le x_k 
\qquad\text{and}\qquad
(x_k-1)^+ \le x_k,
\end{equation*}
the $k$ smallest components of $(\bm q-\bm e_{\sigma_j})^+$ are 
\begin{equation*}
\{x_1, \ldots, x_{k-1}, (x_k-1)^+\},
\end{equation*}
which are the same as the components of
$
(\bm x-\bm e_k)^+=(x_1,\ldots, (x_k-1)^+).
$
Hence \eqref{lem:projection-long-single} follows.

Finally, the equality $\|\bm v(\bm q) \|_0 = k-1$ implies that
\begin{equation*}
q_{\sigma_1}\le\cdots\le q_{\sigma_{k-1}}<q_{\sigma_k}=q_{\sigma_{k+1}}=\cdots=q_{\sigma_n}=:a.
\end{equation*}
Because $q_i\in \Z_+$  for all $i \in \N$, the strict inequality $q_{\sigma_{k-1}}<q_{\sigma_k}$ implies that $a\ge 1$ and $q_{\sigma_{k-1}}\le a-1 = (a-1)^+$. Then the components of
$
\left(\bm q-\sum_{j=k}^n \bm e_{\sigma_j}\right)^+
$
are
\begin{equation*}
\{ q_{\sigma_1},\ldots,q_{\sigma_{k-1}}, (a-1)^+ ,\ldots, (a-1)^+ \},
\end{equation*}
with $n-k+1$ of the components equal to $(a-1)^+$. Since  $q_{\sigma_i} \le q_{\sigma_{k-1}} \le (a-1)^+$ for every $i\le k-1$, the $k$ smallest components of 
$
(\bm q-\sum_{j=k}^n \bm e_{\sigma_j})^+
$
are 
$$\{ q_{\sigma_1},\ldots,q_{\sigma_{k-1}},(a-1)^+ \} = \{x_1, \ldots, x_{k-1}, (x_k - 1)^+\},$$
which are also the components of $(\bm x - \bm e_k)^+$, proving \eqref{lem:projection-long-block}. \hfill \Halmos
\endproof

\subsection{Proof of Lemma \ref{lem:stochasticOrder}}\label{app:order}

A key to the proof of Lemma \ref{lem:stochasticOrder} is the following result, whose proof appears in \S \ref{app:propcouling}. 

\begin{lemma}[GSC-order preservation]\label{lem:GSCpreserve}
Consider $\bm x, \bm z\in \mathcal R(\Z_{+}^k)$ and $i_{\bm x}, i_{\bm z} \in \llbracket 1, k \rrbracket$, with $i_{\bm z} \le i_{\bm x}$. If
$\bm x \le_{\text{GSC}} \bm z$, then %
\begin{equation*}
	\mathcal R\bigl((\bm x - \bm e_{i_{\bm x}})^+\bigr) \le_{\text{GSC}} \mathcal R\bigl((\bm z - \bm e_{i_{\bm z}})^+\bigr).
\end{equation*}
\end{lemma}

\proof{Proof of Lemma \ref{lem:stochasticOrder}.}

We only provide the proof of Assertion (i) of the lemma, since the proof of Assertion (ii) is analogous.
To this end, we first prove that
\begin{equation} \label{eq:GSCorder}
   \bm X^\pi \le_{\text{st,GSC}} \bm X^{\pi^\prime},
\end{equation}
provided that the order holds at time $0$, via a coupling argument. Specifically, we construct random elements $(\bm X,\bm Y)$ and $(\bm X',\bm Y')$, defined jointly on the same probability space, such that $(\bm X,\bm Y)$ has the same distribution as $(\bm X^\pi,\bm Y^\pi)$, while $(\bm X',\bm Y')$ has the same distribution as $(\bm X^{\pi'},\bm Y^{\pi'})$. We then show that the coupling can be constructed so that
\begin{equation*}
\bm X(t)\le_{\text{GSC}} \bm X'(t) \quad \text{w.p.1, for all } t \ge 0.   
\end{equation*}
We refer to $(\bm X,\bm Y)$ as System $\mathcal L$ and to $(\bm X',\bm Y')$ as System $\mathcal U$.

Let an event be either an arrival of a job or a completion of a task from either system.
For $m\ge 1$, let $\mathcal T_m$ denote the $m$th event time, with $\mathcal T_0 := 0$. Take $\bm X(0)$ and $\bm X^\prime(0)$ such that $\bm X(0) \le_{\text{GSC}} \bm X^\prime(0)$ w.p.1, and let $(\bm x, \bm x^\prime, \bm y) := \big(\bm X(\mathcal T_m), \bm X^\prime(\mathcal T_m),  \bm Y(\mathcal T_m) \big)$ for some $m\ge 1$. 

Consider System $\mathcal L$ at time $\mathcal T_m$. The time until the next arrival is exponentially distributed with rate $\la$ and the time until the next task completion from the $i$th shortest queue is exponentially distributed with rate $y_i$. If the $i$th shortest queue is empty, we schedule a ``dummy'' task completion, which does not change the queue length (the queue makes a fictitious transition from the empty state back into the empty state). Throughout the proof, when we refer to the event of a task completion, it includes the dummy task completions. Under this construction, the time until the next task completion from any of the queues is exponentially distributed with rate $\sum_{i=1}^k y_i = 1$. If the $(m+1)$st event is a task completion, we take $I_m$ be the index of the ordered queue vector in which this task completion occurs, so that %
$\mathbb P(I_m = i) = y_i$, $i\in \llbracket 1, k \rrbracket$.  

Next, consider System $\mathcal U$ at time $\mathcal T_m$. As in System $\mathcal L$, the time until the next job arrival in System $\mathcal U$ is exponentially distributed with rate $\la$, and the time until the next task completion is exponentially distributed with rate $1$. If the $(m+1)$st event in System $\mathcal U$ is a task completion, then it occurs at the $I_m^\prime$th shortest queue with probability $\mathbb P(I_m^\prime = i) = b_i$, $i\in \llbracket 1, k \rrbracket$. 

Since we allow dummy transitions, we can act as if all servers in both systems are constantly working at a combined rate $1$, and task completions (including ``dummy completions'') occur according to a unit-rate Poisson process. Further, we can use the same unit-rate Poisson process to generate task-completion epochs simultaneously in both systems, and then independently determine at which queue the task completion occurs, by generating the values of the queue indices $I_m$ and $I_m^\prime$. 

To this end, let $A_m$ be an exponentially distributed random variable with rate $\la$; $S_m$ be exponentially distributed with rate $1$; and $U_m$ be uniformly distributed on $[0,1]$, assuming these three random variables are mutually independent and are independent of all other random variables and processes. We use $A_m$ and $S_m$ to determine, respectively, the time until the next arrival to both systems, and the time until the next task completion in both systems after time $\mathcal T_m$. We use $U_m$ to determine the index of the queue at which the task completion occurs. Then 
\begin{equation*}
\mathcal T_{m+1} = \mathcal T_m + \min\{A_m,S_m\}.
\end{equation*}
If $A_m < S_m$, the next event is an arrival of a job to both systems, so that 
\begin{equation} \label{eq:arrival-event}
	\bm X(\mathcal T_{m+1}) = \bm x+\bm e \qandq \bm X^\prime(\mathcal T_{m+1}) = \bm x^\prime+\bm e.
\end{equation}
If $A_m > S_m$, then the next event is a task completion from the $I_m$th shortest queue in System $\mathcal L$, and a task completion from the $I_m^\prime$th shortest queue in System $\mathcal U$, so that
\begin{equation} \label{eq:completion-event}
	\bm X(\mathcal T_{m+1}) =  \mathcal R \Bigl( (\bm x -\bm e_{I_m})^+ \Bigr), \qquad \bm X^\prime(\mathcal T_{m+1}) =  \mathcal R \Bigl( (\bm x^\prime -\bm e_{I_m^\prime})^+ \Bigr),
\end{equation}
where the operator $(\cdot)^+$ ensures that a dummy transition (due to task completions at an empty queue) does not change the system's state. 

It remains to determine the values of the indices of the queues at which the task completions occur. Let
\begin{equation*}
I_m := \min\left\{j\in \llbracket 1, k\rrbracket:\ \sum_{i=1}^j y_i \ge U_m\right\} \qandq
I_m^\prime := \min\left\{j\in \llbracket 1, k\rrbracket:\ \sum_{i=1}^j b_i \ge U_m\right\}.
\end{equation*}
For $i\in \llbracket 1, k \rrbracket$, let $C_i : = \sum_{j=1}^i y_j$, with $C_0 := 0$. Observe that $I_m = i$ if and only if $C_{i-1} < U_m \le C_i$ for any $i\in \llbracket 1, k\rrbracket$, and thus $\mathbb P(I_m = i) = \mathbb P( C_{i-1} < U_m \le C_i ) = C_i - C_{i-1} = y_i$. Similarly, $\mathbb P(I_m^\prime = i) = b_i$. Since $\bm Y^\prime = \bm b \eta$ and $\bm Y(\mathcal T_m) = \bm y$, the ordering $\bm Y \ge_{\text{GSC}} \bm Y^\prime$ implies that $\bm y \ge_{\text{GSC}} \bm b$, which in turn yields 
\begin{equation} \label{eq:index-order}
	I_m \ge I_m^\prime ~\text{w.p.1}.
\end{equation}

We next prove by induction that
\begin{equation} \label{IneqX}
\bm X(\mathcal T_m)\le_{\text{GSC}}\bm X^\prime(\mathcal T_m) \quad \text{w.p.1 for all } m\ge0.
\end{equation}
By construction, \eqref{IneqX} holds for $m=0$. Assume now that it holds at time $\mathcal T_m$ for some $m > 0$.

If the $(m+1)$st event is a job arrival in both systems, then \eqref{eq:arrival-event} implies that \eqref{IneqX} also holds at time $\mathcal T_{m+1}$. If the $(m+1)$st event is a task completion in both systems, then the states at time $\mathcal T_{m+1}$ are given by \eqref{eq:completion-event}. Since $\bm X(\mathcal T_m)\le_{\text{GSC}}\bm X^\prime(\mathcal T_m)$ by the induction hypothesis and $I_m\ge I_m^\prime$ w.p.1 by \eqref{eq:index-order}, Lemma \ref{lem:GSCpreserve} yields
$$\bm X(\mathcal T_{m+1})\le_{\text{GSC}}\bm X^\prime(\mathcal T_{m+1}).$$
Thus, \eqref{IneqX} holds at time $\mathcal T_{m+1}$, completing the induction.
Since both processes are constant between events, and $\mathcal T_m\uparrow\infty$ as $m \uparrow \infty$ w.p.1., we have that
$\bm X(t)\le_{\text{GSC}}\bm X^\prime(t)$ w.p.1 %
for all $t \ge 0$, from which \eqref{eq:GSCorder} follows.

Finally, $X_k^\pi(t) = Z^\pi(t)$ and $X_k^{\pi^\prime}(t) = Z^{\pi^\prime}(t)$ by \eqref{eq:identity-z-x}, and the stochastic order relation $\bm X^{\pi} \le_{\text{st,GSC}} \bm  X^{\pi^\prime}$ implies that $Z^{\pi} \le_{\text{st}} Z^{\pi^\prime}$. Further, $\bm Q^{\pi^\prime}(t) = \bm 0$ if and only if $Z^{\pi^\prime}(t) = 0$, and $\bm Q^{\pi}(t) = \bm 0$ if and only if $Z^{\pi}(t) = 0$.
Now, $\bm Q^{\pi^\prime}$ is positive recurrent for any $\rho \in \mathcal S^{(n,k)} (\pi^\prime)$, so that $0$ is a regeneration state for $Z^{\pi^\prime}$ with finite expected regeneration time. It follows that $Z^\pi$ is also a positive recurrent regenerative process, and therefore $\bm Q^\pi$ is positive recurrent. We conclude that $\mathcal S^{(n,k)} (\pi^\prime) \subseteq \mathcal S^{(n,k)} (\pi)$, as stated. \hfill \Halmos

\endproof

\subsection{Proof of Lemma \ref{lem:GSCpreserve} } \label{app:propcouling}

\proof{Proof.}
Fix $m\in\llbracket 1,k\rrbracket$, and let 
\begin{equation*}
T_m(\bm u):=\sum_{r=m}^k u_r,
\quad \bm u\in\mathcal R(\Z_+^k), \qandq T_{k+1}(\bm u):=0.
\end{equation*}
The condition in the statement of the lemma that $\bm x\le_{\text{GSC}}\bm z$ is equivalent to
\begin{equation}\label{eq:gscpreserve-tail-order}
T_m(\bm z)-T_m(\bm x)\ge 0, ~ m\in\llbracket 1,k\rrbracket.
\end{equation}
To simplify the notation, let
\begin{equation} \label{xzstar}
\bm x^*:=\mathcal R\bigl((\bm x-\bm e_{i_{\bm x}})^+\bigr)
\qandq
\bm z^*:=\mathcal R\bigl((\bm z-\bm e_{i_{\bm z}})^+\bigr).
\end{equation}
For $j_{\bm x}:=\min\{r\in\llbracket 1,i_{\bm x}\rrbracket:x_r=x_{i_{\bm x}}\}$, it holds that $x_{j_{\bm x}} = x_{j_{\bm x}+1 } = \cdots = x_{i_{\bm x}}$. Hence, $\bm x^*$ is obtained from $\bm x$ by replacing the $j_{\bm x}$th component $x_{j_{\bm x}}$ by $(x_{j_{\bm x}} - 1)^+$. Therefore,
\begin{equation*}
	T_m(\bm x^*) = T_m(\bm x) - \mathbbm{1}\{ x_{i_{\bm x}} > 0, m\le j_{\bm x} \}.
\end{equation*}
Similarly, for
	$j_{\bm z}:=\min\{r\in\llbracket 1,i_{\bm z}\rrbracket:z_r=z_{i_{\bm z}}\}$,
it holds that
\begin{equation*}
	T_m(\bm z^*) = T_m(\bm z) - \mathbbm{1}\{ z_{i_{\bm z}} > 0, m\le j_{\bm z} \}. %
\end{equation*}
Consequently, for every $m\in\llbracket 1,k\rrbracket$,
\begin{equation}\label{eq:gscpreserve-main-diff}
T_m(\bm z^*)-T_m(\bm x^*)
= \big(T_m(\bm z)-T_m(\bm x) \big)
- \big(\mathbbm{1}\{z_{i_{\bm z}}>0,\ m\le j_{\bm z}\}
 - \mathbbm{1}\{x_{i_{\bm x}}>0,\ m\le j_{\bm x}\} \big).
\end{equation}

We next prove that $T_m(\bm z^*) - T_m(\bm x^*) \ge 0$ for every $m\in\llbracket 1,k\rrbracket$. To this end, observe that if
\begin{equation*}
\mathbbm{1}\{z_{i_{\bm z}}>0,\ m\le j_{\bm z}\}
\le
\mathbbm{1}\{x_{i_{\bm x}}>0,\ m\le j_{\bm x}\},
\end{equation*}
then \eqref{eq:gscpreserve-tail-order} and \eqref{eq:gscpreserve-main-diff} imply that
$T_m(\bm z^*)-T_m(\bm x^*)\ge0$. Thus, we need to show that $T_m(\bm z) - T_m(\bm x) \ge 1$ whenever 
\begin{equation} \label{eq:indicator}
	\mathbbm{1}\{z_{i_{\bm z}}>0,\ m\le j_{\bm z}\} = 1 \quad \text{and} \quad  \mathbbm{1}\{x_{i_{\bm x}}>0,\ m\le j_{\bm x}\} = 0.
\end{equation}

Observe that \eqref{eq:indicator} holds if and only if one of the following two mutually exclusive cases hold:
\begin{align*}
	& \text{Case 1: } ~z_{i_{\bm z}}>0,
\quad m\le j_{\bm z},
\quad x_{i_{\bm x}}=0; \\
	& \text{Case 2: } ~z_{i_{\bm z}}>0,
\quad j_{\bm x}<m\le j_{\bm z},
\quad x_{i_{\bm x}}>0. 
\end{align*}
We thus show $T_m(\bm z) - T_m(\bm x) \ge 1$ in either of these two cases. 

\paragraph{{\bf Proof for Case 1.}} %

Since $m\le j_{\bm z}\le i_{\bm z}\le i_{\bm x}$ and $\bm x$ is nonnegative and nondecreasing, $x_r=0$ for all $r\in\llbracket m,j_{\bm z}\rrbracket$. Hence
\begin{equation}\label{eq:gscpreserve-case1-tail}
\begin{aligned}
T_m(\bm z)-T_m(\bm x) =T_{j_{\bm z}+1}(\bm z)-T_{j_{\bm z}+1}(\bm x)
  +\sum_{r=m}^{j_{\bm z}}(z_r-x_r) \ge z_{j_{\bm z}}=z_{i_{\bm z}}\ge1,
\end{aligned}
\end{equation}
where the first inequality follows from \eqref{eq:gscpreserve-tail-order} and $z_r \ge x_r=0$ for $r \in \llbracket m,j_{\bm z}\rrbracket$, and the equality $z_{j_{\bm z}}=z_{i_{\bm z}}$ follows from the definition of $j_{\bm z}$. 

\paragraph{{\bf Proof of Case 2.}} %
In this case, it holds that
\begin{equation} \label{eq:index}
	j_{\bm x}<j_{\bm z}\le i_{\bm z}\le i_{\bm x}.
\end{equation} 
We first prove that
\begin{equation}\label{eq:gscpreserve-jz-tail-one}
T_{j_{\bm z}}(\bm z)-T_{j_{\bm z}}(\bm x)\ge 1,
\end{equation}
by taking the assumption that  \eqref{eq:gscpreserve-jz-tail-one} does not hold and arriving at a contradiction.
In particular, suppose that $T_{j_{\bm z}}(\bm z)-T_{j_{\bm z}}(\bm x) = 0$. Then 
\begin{equation}\label{eq:gscpreserve-case2-zero}
0 =T_{j_{\bm z}}(\bm z)-T_{j_{\bm z}}(\bm x) =T_{j_{\bm z}+1}(\bm z)-T_{j_{\bm z}+1}(\bm x)+z_{j_{\bm z}}-x_{j_{\bm z}}.
\end{equation}
Because $T_{j_{\bm z}+1}(\bm z)-T_{j_{\bm z}+1}(\bm x)\ge0$ by \eqref{eq:gscpreserve-tail-order}, and $x_{i_{\bm x}} = x_{i_{\bm z}}$ due to \eqref{eq:index}, it follows from \eqref{eq:gscpreserve-case2-zero} that
\begin{equation}\label{eq:gscpreserve-xge-z}
x_{i_{\bm x}} = x_{i_{\bm z}} \ge z_{i_{\bm z}}.
\end{equation}
For every $r\in\llbracket j_{\bm x},j_{\bm z}-1\rrbracket$, it holds that $x_r=x_{i_{\bm x}}$ by the definition of $j_{\bm x}$, and $z_r< z_{j_{\bm z}} = z_{i_{\bm z}}\le x_{i_{\bm x}}$ by the definition of $j_{\bm z}$ together with \eqref{eq:gscpreserve-xge-z}. Hence, it follows from \eqref{eq:gscpreserve-case2-zero} that
\begin{equation*}
\begin{aligned}
T_{j_{\bm x}}(\bm z)-T_{j_{\bm x}}(\bm x)
=T_{j_{\bm z}}(\bm z)-T_{j_{\bm z}}(\bm x)
  +\sum_{r=j_{\bm x}}^{j_{\bm z}-1}(z_r-x_r)
=\sum_{r=j_{\bm x}}^{j_{\bm z}-1}(z_r-x_{i_{\bm x}})<0,
\end{aligned}
\end{equation*}
contradicting \eqref{eq:gscpreserve-tail-order}. We therefore conclude that \eqref{eq:gscpreserve-jz-tail-one} must hold.

We next prove that $T_r(\bm z)-T_r(\bm x) \ge 1$ for any $r\in \llbracket j_{\bm x}+1,j_{\bm z}-1\rrbracket$ by assuming, for the sake of contradiction, that $T_r(\bm z)-T_r(\bm x) = 0$ for at least one index $r\in \llbracket j_{\bm x}+1,j_{\bm z}-1\rrbracket$. Let $r_0$ be the largest such index in $\llbracket j_{\bm x}+1,j_{\bm z}-1\rrbracket$. Then
\begin{equation}\label{eq:gscpreserve-r0-plus}
T_{r_0+1}(\bm z)-T_{r_0+1}(\bm x)\ge1,
\end{equation}
and since $r_0 \in \llbracket j_{\bm x}+1,j_{\bm z}-1\rrbracket \subseteq \llbracket j_{\bm x}, i_{\bm x} \rrbracket$, we have $x_{r_0-1} = x_{r_0}=x_{i_{\bm x}}$. Therefore,
\begin{align*}
z_{r_0-1}-x_{r_0-1} \le 
z_{r_0}-x_{r_0 - 1} = z_{r_0} - x_{r_0}
=\bigl(T_{r_0}(\bm z)-T_{r_0}(\bm x)\bigr)
 -\bigl(T_{r_0+1}(\bm z)-T_{r_0+1}(\bm x)\bigr) 
\le -1,
\end{align*}
where the last inequality follows from \eqref{eq:gscpreserve-r0-plus} because $T_{r_0}(\bm z)-T_{r_0}(\bm x)=0$. 
Thus
\begin{equation*}
\begin{aligned}
T_{r_0-1}(\bm z)-T_{r_0-1}(\bm x) =T_{r_0}(\bm z)-T_{r_0}(\bm x)+z_{r_0-1}-x_{r_0-1}
= z_{r_0-1}-x_{r_0-1} \le -1,
\end{aligned}
\end{equation*}
contradicting \eqref{eq:gscpreserve-tail-order}. 
Since we arrive at a contradiction, we conclude that $T_r(\bm z)-T_r(\bm x) \ge 1$ for any $r\in \llbracket j_{\bm x}+1,j_{\bm z}-1\rrbracket$, which, together with \eqref{eq:gscpreserve-jz-tail-one}, implies that $T_m(\bm z)-T_m(\bm x)\ge 1$ in Case 2.

Finally, since $T_m(\bm z) - T_m(\bm x)\ge 1$ for all $m\in\llbracket 1,k\rrbracket$, we have that $T_m(\bm x^*)\le T_m(\bm z^*)$, so that $\bm x^*\le_{\text{GSC}}\bm z^*$, for $\bm x^*$ and $\bm z^*$ in \eqref{xzstar}, as stated. \hfill \Halmos
\endproof

\section*{Acknowledgments}
This material is based upon work supported by the National Science Foundation under Award No. 1826353. Any opinions, findings and conclusions or recommendations expressed in this material are those of the authors and do not necessarily reflect the views of the National Science Foundation.

\bibliographystyle{informs2014}
\bibliography{bibdocument} %

@misc{nvidia_dgx_superpod_2026,
  author       = {{NVIDIA}},
  title        = {{NVIDIA DGX SuperPOD}},
  year         = {2026},
  howpublished = {\url{https://www.nvidia.com/en-us/data-center/dgx-superpod/}},
  note         = {Accessed July 12, 2026}
}

@article{kim2026atropos,
  title={Atropos: Improving Cost-Benefit Trade-off of LLM-based Agents under Self-Consistency with Early Termination and Model Hotswap},
  author={Kim, Naryeong and Yoo, Shin},
  journal={arXiv preprint arXiv:2604.15075},
  year={2026}
}

@article{agarwal2025first,
  title={First finish search: Efficient test-time scaling in large language models},
  author={Agarwal, Aradhye and Sengupta, Ayan and Chakraborty, Tanmoy},
  journal={arXiv preprint arXiv:2505.18149},
  year={2025}
}

@book{dai2020processing,
  title={Processing networks: fluid models and stability},
  author={Dai, JG and Harrison, J Michael},
  year={2020},
  publisher={Cambridge University Press}
}

@article{moyal2022stability,
  title={Stability of parallel server systems},
  author={Moyal, Pascal and Perry, Ohad},
  journal={Operations Research},
  volume={70},
  number={4},
  pages={2456--2476},
  year={2022},
  publisher={INFORMS}
}

@article{dai2024deepseekmoe,
  title   = {DeepSeekMoE: Towards Ultimate Expert Specialization in Mixture-of-Experts Language Models},
  author  = {Dai, Damai and Deng, Chengqi and Zhao, Chenggang and Xu, R. X. and Gao, Huazuo and Chen, Deli and Li, Jiashi and Zeng, Wangding and Yu, Xingkai and Wu, Y. and Xie, Zhenda and Li, Y. K. and Huang, Panpan and Luo, Fuli and Ruan, Chong and Sui, Zhifang and Liang, Wenfeng},
  journal = {arXiv preprint arXiv:2401.06066},
  year    = {2024}
}

@article{meijer2024optimization,
  title={Optimization of Inventory and Capacity in Large-Scale Assembly Systems Using Extreme-Value Theory},
  author={Meijer, Mirjam S and Schol, Dennis and van Jaarsveld, Willem and Vlasiou, Maria and Zwart, Bert},
  journal={Stochastic Systems},
  year={2024},
  publisher={INFORMS}
}

@article{schol2022large,
  title={Large fork-join queues with nearly deterministic arrival and service times},
  author={Schol, Dennis and Vlasiou, Maria and Zwart, Bert},
  journal={Mathematics of Operations Research},
  volume={47},
  number={2},
  pages={1335--1364},
  year={2022},
  publisher={INFORMS}
}

@article{harchol2021open,
  title={Open problems in queueing theory inspired by datacenter computing},
  author={Harchol-Balter, Mor},
  journal={Queueing Systems},
  volume={97},
  number={1},
  pages={3--37},
  year={2021},
  publisher={Springer}
}

@article{wang2019delay,
  title={Delay asymptotics and bounds for multi-task parallel jobs},
  author={Wang, Weina and Harchol-Balter, Mor and Jiang, Haotian and Scheller-Wolf, Alan and Srikant, Rayadurgam},
  journal={ACM SIGMETRICS Performance Evaluation Review},
  volume={46},
  number={3},
  pages={2--7},
  year={2019},
  publisher={ACM New York, NY, USA}
}

@article{andradottir2005throughput,
  title={Throughput maximization for tandem lines with two stations and flexible servers},
  author={Andrad{\'o}ttir, Sigr{\'u}n and Ayhan, Hayriye},
  journal={Operations Research},
  volume={53},
  number={3},
  pages={516--531},
  year={2005},
  publisher={INFORMS}
}

@article{li2020federated,
  title={Federated learning: Challenges, methods, and future directions},
  author={Li, Tian and Sahu, Anit Kumar and Talwalkar, Ameet and Smith, Virginia},
  journal={IEEE signal processing magazine},
  volume={37},
  number={3},
  pages={50--60},
  year={2020},
  publisher={IEEE}
}

@article{lee2017speeding,
  title={Speeding up distributed machine learning using codes},
  author={Lee, Kangwook and Lam, Maximilian and Pedarsani, Ramtin and Papailiopoulos, Dimitris and Ramchandran, Kannan},
  journal={IEEE Transactions on Information Theory},
  volume={64},
  number={3},
  pages={1514--1529},
  year={2017},
  publisher={IEEE}
}

@article{hu2021distributed,
  title={Distributed machine learning for wireless communication networks: Techniques, architectures, and applications},
  author={Hu, Shuyan and Chen, Xiaojing and Ni, Wei and Hossain, Ekram and Wang, Xin},
  journal={IEEE Communications Surveys \& Tutorials},
  volume={23},
  number={3},
  pages={1458--1493},
  year={2021},
  publisher={IEEE}
}

@article{ouyang2022training,
  title={Training language models to follow instructions with human feedback},
  author={Ouyang, Long and Wu, Jeffrey and Jiang, Xu and Almeida, Diogo and Wainwright, Carroll and Mishkin, Pamela and Zhang, Chong and Agarwal, Sandhini and Slama, Katarina and Ray, Alex and others},
  journal={Advances in neural information processing systems},
  volume={35},
  pages={27730--27744},
  year={2022}
}

@inproceedings{marin2016dynamic,
  title={Dynamic control of the join-queue lengths in saturated fork-join stations},
  author={Marin, Andrea and Rossi, Sabina},
  booktitle={International Conference on Quantitative Evaluation of Systems},
  pages={123--138},
  year={2016},
  organization={Springer}
}

@article{joshi2017efficient,
  title={Efficient redundancy techniques for latency reduction in cloud systems},
  author={Joshi, Gauri and Soljanin, Emina and Wornell, Gregory},
  journal={ACM Transactions on Modeling and Performance Evaluation of Computing Systems (TOMPECS)},
  volume={2},
  number={2},
  pages={1--30},
  year={2017},
  publisher={ACM New York, NY, USA}
}

@article{ko2008sojourn,
  title={Sojourn times in G/M/1 fork-join networks},
  author={Ko, Sung-Seok and Serfozo, Richard F},
  journal={Naval Research Logistics (NRL)},
  volume={55},
  number={5},
  pages={432--443},
  year={2008},
  publisher={Wiley Online Library}
}

@article{varki1999mean,
  title={Mean value technique for closed fork-join networks},
  author={Varki, Elizabeth},
  journal={ACM SIGMETRICS Performance Evaluation Review},
  volume={27},
  number={1},
  pages={103--112},
  year={1999},
  publisher={ACM New York, NY, USA}
}

@article{nelson1988approximate,
  title={Approximate analysis of fork/join synchronization in parallel queues},
  author={Nelson, Randolph and Tantawi, Asser N},
  journal={IEEE transactions on computers},
  volume={37},
  number={6},
  pages={739--743},
  year={1988},
  publisher={IEEE}
}

@article{flatto1984two,
  title={Two parallel queues created by arrivals with two demands I},
  author={Flatto, Leopold and Hahn, Sann},
  journal={SIAM Journal on Applied Mathematics},
  volume={44},
  number={5},
  pages={1041--1053},
  year={1984},
  publisher={SIAM}
}

@article{baccelli1989fork,
  title={The fork-join queue and related systems with synchronization constraints: Stochastic ordering and computable bounds},
  author={Baccelli, Francois and Makowski, Armand M and Shwartz, Adam},
  journal={Advances in Applied Probability},
  volume={21},
  number={3},
  pages={629--660},
  year={1989},
  publisher={Cambridge University Press}
}

@article{down2006dynamic,
  title={Dynamic load balancing in parallel queueing systems: Stability and optimal control},
  author={Down, Douglas G and Lewis, Mark E},
  journal={European Journal of Operational Research},
  volume={168},
  number={2},
  pages={509--519},
  year={2006},
  publisher={Elsevier}
}

@article{thomasian2014analysis,
  title={Analysis of fork/join and related queueing systems},
  author={Thomasian, Alexander},
  journal={ACM Computing Surveys (CSUR)},
  volume={47},
  number={2},
  pages={1--71},
  year={2014},
  publisher={ACM New York, NY, USA}
}

@inproceedings{pedarsani2014scheduling,
  title={Scheduling tasks with precedence constraints on multiple servers},
  author={Pedarsani, Ramtin and Walrand, Jean and Zhong, Yuan},
  booktitle={2014 52nd Annual Allerton Conference on Communication, Control, and Computing (Allerton)},
  pages={1196--1203},
  year={2014},
  organization={IEEE}
}

@inproceedings{pedarsani2014robust,
  title={Robust scheduling in a flexible fork-join network},
  author={Pedarsani, Ramtin and Walrand, Jean and Zhong, Yuan},
  booktitle={53rd IEEE Conference on Decision and Control},
  pages={3669--3676},
  year={2014},
  organization={IEEE}
}

@article{andradottir2003dynamic,
  title={Dynamic server allocation for queueing networks with flexible servers},
  author={Andrad{\'o}ttir, Sigr{\'u}n and Ayhan, Hayriye and Down, Douglas G},
  journal={Operations Research},
  volume={51},
  number={6},
  pages={952--968},
  year={2003},
  publisher={INFORMS}
}

@inproceedings{joshi2012coding,
  title={Coding for fast content download},
  author={Joshi, Gauri and Liu, Yanpei and Soljanin, Emina},
  booktitle={2012 50th Annual Allerton Conference on Communication, Control, and Computing (Allerton)},
  pages={326--333},
  year={2012},
  organization={IEEE}
}

@article{joshi2014delay,
  title={On the delay-storage trade-off in content download from coded distributed storage systems},
  author={Joshi, Gauri and Liu, Yanpei and Soljanin, Emina},
  journal={IEEE Journal on Selected Areas in Communications},
  volume={32},
  number={5},
  pages={989--997},
  year={2014},
  publisher={IEEE}
}

@article{gardner2016queueing,
  title={Queueing with redundant requests: exact analysis},
  author={Gardner, Kristen and Zbarsky, Samuel and Doroudi, Sherwin and Harchol-Balter, Mor and Hyyti{\"a}, Esa and Scheller-Wolf, Alan},
  journal={Queueing Systems},
  volume={83},
  number={3},
  pages={227--259},
  year={2016},
  publisher={Springer}
}

@article{gardner2017redundancy,
  title={Redundancy-d: The power of d choices for redundancy},
  author={Gardner, Kristen and Harchol-Balter, Mor and Scheller-Wolf, Alan and Velednitsky, Mark and Zbarsky, Samuel},
  journal={Operations Research},
  volume={65},
  number={4},
  pages={1078--1094},
  year={2017},
  publisher={INFORMS}
}

@article{bassamboo2012little,
  title={A little flexibility is all you need: on the asymptotic value of flexible capacity in parallel queuing systems},
  author={Bassamboo, Achal and Randhawa, Ramandeep S and Mieghem, Jan A Van},
  journal={Operations Research},
  volume={60},
  number={6},
  pages={1423--1435},
  year={2012},
  publisher={INFORMS}
}

@article{anton2021stability,
  title={On the stability of redundancy models},
  author={Anton, Elene and Ayesta, Urtzi and Jonckheere, Matthieu and Verloop, Ina Maria},
  journal={Operations Research},
  volume={69},
  number={5},
  pages={1540--1565},
  year={2021},
  publisher={INFORMS}
}

@article{gardner2019smart,
  title={Smart dispatching in heterogeneous systems},
  author={Gardner, Kristen and Stephens, Cole},
  journal={ACM SIGMETRICS Performance Evaluation Review},
  volume={47},
  number={2},
  pages={12--14},
  year={2019},
  publisher={ACM New York, NY, USA}
}

@book{sethuraman2022analysis,
  title={Analysis of fork-join systems: Network of queues with precedence constraints},
  author={Sethuraman, Samyukta},
  year={2022},
  publisher={CRC Press}
}

@article{rizk2016stochastic,
  title={Stochastic bounds in fork--join queueing systems under full and partial mapping},
  author={Rizk, Amr and Poloczek, Felix and Ciucu, Florin},
  journal={Queueing Systems},
  volume={83},
  number={3-4},
  pages={261--291},
  year={2016},
  publisher={Springer}
}

@article{carmeli2023state,
  title={State-dependent estimation of delay distributions in fork-join networks},
  author={Carmeli, Nitzan and Yom-Tov, Galit B and Boxma, Onno J},
  journal={Manufacturing \& Service Operations Management},
  volume={25},
  number={3},
  pages={1081--1098},
  year={2023},
  publisher={INFORMS}
}

@article{ozkan2019control,
  title={On the control of fork-join networks},
  author={{\"O}zkan, Erhun and Ward, Amy R},
  journal={Mathematics of Operations Research},
  volume={44},
  number={2},
  pages={532--564},
  year={2019},
  publisher={INFORMS}
}

@article{ozkan2022control,
  title={Control of Fork-Join Processing Networks with Multiple Job Types and Parallel Shared Resources},
  author={{\"O}zkan, Erhun},
  journal={Mathematics of Operations Research},
  volume={47},
  number={2},
  pages={1310--1334},
  year={2022},
  publisher={INFORMS}
}

@article{joshi2015queues,
  title={Queues with redundancy: Latency-cost analysis},
  author={Joshi, Gauri and Soljanin, Emina and Wornell, Gregory},
  journal={ACM SIGMETRICS Performance Evaluation Review},
  volume={43},
  number={2},
  pages={54--56},
  year={2015},
  publisher={ACM New York, NY, USA}
}

@article{shah2015redundant,
  title={When do redundant requests reduce latency?},
  author={Shah, Nihar B and Lee, Kangwook and Ramchandran, Kannan},
  journal={IEEE Transactions on Communications},
  volume={64},
  number={2},
  pages={715--722},
  year={2015},
  publisher={IEEE}
}

@article{lee2017mds,
  title={The MDS queue: Analysing the latency performance of erasure codes},
  author={Lee, Kangwook and Shah, Nihar B and Huang, Longbo and Ramchandran, Kannan},
  journal={IEEE Transactions on Information Theory},
  volume={63},
  number={5},
  pages={2822--2842},
  year={2017},
  publisher={IEEE}
}

@misc{anthropic_agents_2024,
  author       = {Anthropic},
  title        = {Building effective agents},
  year         = {2024},
  url          = {https://www.anthropic.com/research/building-effective-agents},
  note         = {Accessed: February 2, 2025}
}

@article{pedarsani2017robust,
  title={Robust scheduling for flexible processing networks},
  author={Pedarsani, Ramtin and Walrand, Jean and Zhong, Yuan},
  journal={Advances in Applied Probability},
  volume={49},
  number={2},
  pages={603--628},
  year={2017},
  publisher={Cambridge University Press}
}

\end{document}